\documentclass[aps,pra,twocolumn,superscriptaddress]{revtex4}
\usepackage{amsthm}
\usepackage{amsmath}
\usepackage{latexsym}
\usepackage{amsfonts}
\usepackage{amssymb}
\usepackage{color}
\usepackage{bbm,dsfont}
\usepackage{graphicx}
\usepackage{hyperref}

\usepackage{xr}

\newtheorem{proposition}{Proposition}
\newtheorem{appendixproposition}{Proposition}[section]
\newtheorem{theorem}{Theorem}
\newtheorem{appendixtheorem}{Theorem}[section]
\newtheorem{appendixlemma}{Lemma}[section]
\newtheorem{corollary}{Corollary}

\theoremstyle{definition}

\newtheorem{definition}{Definition}

\newcommand{\hi}{\mathcal{H}} 
\newcommand{\ki}{\mathcal{K}} 
\newcommand{\ket}[1]{|#1\rangle} 
\newcommand{\kb}[2]{|#1\rangle\langle#2|} 
\newcommand{\tr}[1]{\textrm{tr}\left[#1\right]} 

\newcommand{\id}{\mathbbm{1}} 

\newcommand{\A}{\mathsf{A}}
\newcommand{\B}{\mathsf{B}}
\newcommand{\F}{\mathsf{F}}
\newcommand{\G}{\mathsf{G}}
\newcommand{\Q}{\mathsf{Q}}
\renewcommand{\P}{\mathsf{P}}
\newcommand{\M}{\mathsf{M}}
\newcommand{\I}{\mathsf{I}}

\newcommand{\coh}{{\rm coh}}
\newcommand{\inv}{\mathfrak S}

\begin{document}
\title[]{Amount of quantum coherence needed for measurement incompatibility}
\author{Jukka Kiukas}
\affiliation{Department of Mathematics, Aberystwyth University, Aberystwyth, SY23 3BZ, U.K.}
\author{Daniel McNulty}
\affiliation{Department of Mathematics, Aberystwyth University, Aberystwyth, SY23 3BZ, U.K.}
\author{Juha-Pekka Pellonp\"a\"a}
\affiliation{Department of Physics and Astronomy, University of Turku, FI-20014 Turun yliopisto, Finland}
\begin{abstract} A pair of quantum observables diagonal in the same ``incoherent'' basis can be measured jointly, so \emph{some} coherence is obviously required for measurement incompatibility. Here we first observe that coherence in a single observable is linked to the diagonal elements of any observable jointly measurable with it, leading to a general criterion for the coherence needed for incompatibility. Specialising to the case where the second observable is incoherent (diagonal), we develop a concrete method for solving incompatibility problems, tractable even in large systems by analytical bounds, without resorting to numerical optimisation.
We verify the consistency of our method by a quick proof of the known noise bound for mutually unbiased bases, and apply it to study emergent classicality in the spin-boson model of an $N$-qubit open quantum system. Finally, we formulate our theory in an operational resource-theoretic setting involving ``genuinely incoherent operations'' used previously in the literature, and show that if the coherence is insufficient to sustain incompatibility, the associated joint measurements have sequential implementations via incoherent instruments. 
\end{abstract}
\maketitle

\section{Introduction}
Coherence typically refers to nonzero off-diagonal elements in a quantum state, and is an essential resource for quantum information tasks \cite{baumgratz14,streltsov17,winter16,devicente17, chitambar19}. Coherence in \emph{measurements} (observables) \cite{oszmaniec19, baek20} is equally fundamental, with an obvious relation to the non-commutativity of projective measurements, which has recently been refined \cite{styliaris19, cimini19,bishof19}. It is therefore natural to ask how coherence is related to incompatibility of general observables -- positive operator valued measures (POVMs). Incompatibility is a resource as well \cite{heinosaari15a,chitambar19}, specifically for steering \cite{wolf09, uola15,quintino14,heinosaari15b,kiukas17, kiukas16} and state discrimination \cite{skrzypczyk19,carmeli18, carmeli19a}, and clearly requires non-commutativity, hence coherence.

As usual \cite{streltsov17}, we define coherence relative to a fixed ``incoherent'' basis (there is also a basis-independent approach \cite{designolle20}). Our key observation is the following: while incompatibility of POVMs is not linked to the overall coherence in their matrices, there is an asymmetric entry-wise relation: coherences in \emph{one} POVM are linked to the corresponding diagonal probabilities of any POVM jointly measurable with it. Heuristically, \emph{an observable that sharply distinguishes a pair of basis elements is incompatible with observables detecting coherence between that pair}. An extreme case is any basis observable mutually unbiased \cite{durt10} to the incoherent basis -- it is both complementary and maximally coherent.

We warm up in section \ref{basic} by formalising the above observation into a simple but completely general inequality, the violation of which witnesses incompatibility. Combined with a sufficient condition for incompatibility, this leads to an analytical method for tackling the incompatibility problem, generalising the usual robustness idea \cite{designolle19b,bluhm20}, and easily reproducing the known noise bound for incompatible mutually unbiased bases (MUB) \cite{carmeli12,uola16,designolle19a,carmeli19a,carmeli19b}. In section \ref{cohmatrix} we specialise to the physically motivated setting where measurement coherence is given by a fixed ``pattern matrix'' describing decoherence \cite{breuer02, buscemi05,kayser15}, and subsequently use it to study emergent classicality in the spin-boson model \cite{breuer02,unruh95,palma96}, including the role of decoherence-free subspaces \cite{palma96,lidar14, lidar98, bacon00}.
Unlike existing results on incompatibility in open systems \cite{addis16}, our method works for arbitrary system size. Finally, in Section \ref{operational}, we formulate the idea in general operational terms motivated by resource theory, including \emph{genuinely incoherent operations} \cite{devicente17,helm09,yao17} and introducing \emph{incoherent instruments}, which turn out to provide sequential implementations for any joint measurement in an instance of channel-observable compatibility \cite{heinosaari13,heinosaari14,heinosaari18}. 

\section{General formulation}\label{basic}
Let $\mathcal H$ be a Hilbert space of $\dim \mathcal H=d<\infty$, and $\{|n\rangle\}_{n=1}^d$ its \emph{incoherent basis} \cite{streltsov17}. An \emph{observable} (POVM) $\M$ with a finite outcome set $\Omega$ consists of positive semidefinite (PSD) matrices $\M(i)\geq 0$, for which $\sum_{i\in \Omega} \M(i)=\id$ (the identity matrix). For any POVM $\M$ we define the \emph{entry-wise coherence} $$\coh_{nm}(\M) := \sum_{i\in \Omega} |\langle n|\M(i)|m\rangle|,\text{ for each }n,m.$$ We note that $0\leq \coh_{nm}(\M)\leq 1$, and call $\M$ \emph{maximally coherent} if $\coh_{nm}(\M)= 1$ for all $n\neq m$. We now observe (Appendix \ref{A:coh}) that any maximally coherent $\M$ with $d$ outcomes is mutually unbiased to the incoherent basis $\{|n\rangle\}$, i.e. $\M(i) =|\psi_i\rangle\langle\psi_i|$ with $|\langle\psi_i|n\rangle|^2=d^{-1}$ for all $n,i$. This reflects the importance of MUBs in the context of measurement coherence.

For any $\M$ we let $p^\M_{n}(j) := \langle n|\M(j)|n\rangle$ be the outcome distribution in state $|n\rangle$. The ability of $\M$ to distinguish $|n\rangle$ from $|m\rangle$ can be quantified by $f$-divergences \cite{csiszar04} between $p^\M_{n}$ and $p^\M_{m}$; we use the \emph{Hellinger distances} \cite{pollard02}
$$d_{nm}^2(\M) :=1-\sum_{j} \sqrt{p^\M_{n}(j)p^\M_{m}(j)}.$$
Finally, an observable $\M$ is \emph{jointly measurable} with an observable $\F$, if there is a \emph{joint observable} $\G=(\G(i,j))_{(i,j)}$ with $\sum_j \G(i,j)=\M(i)$ for all $i$, and $\sum_i \G(i,j)=\F(j)$ for all $j$; otherwise $\M$ and $\F$ are \emph{incompatible} \cite{QM}.

\subsection{Joint measurability criteria}
The following observation provides a simple tradeoff between distinguishability and coherence, under the assumption of joint measurability:
\begin{proposition} \label{basicprop}
If $\M$ and $\F$ are jointly measurable, then
\begin{equation}\label{tradeoff}
\coh_{nm}(\M) +d_{nm}^2(\F)\leq 1 \text{ for all } n,m.
\end{equation}
\end{proposition}
\begin{proof}
We have ${\rm coh}_{nm}(\M)\leq \sum_{i,j} |\langle n|\G(i,j)|m \rangle|$ for any joint POVM $\G$ of $\M$ and $\F$. But $\G(i,j)$ is PSD, and hence $|\langle n|\G(i,j)|m \rangle|\leq \sqrt{p_n^\G(i,j)p_m^\G(i,j)}$ so $\sum_i |\langle n|\G(i,j)|m \rangle|\leq \sqrt{\sum_i p_n^\G(i,j)\sum_i p_m^\G(i,j)}=\sqrt{p_n^\F(j)p_m^\F(j)}$ by the Schwarz inequality.
\end{proof}
Hence, $\M$ and $\F$ are incompatible if \eqref{tradeoff} is violated for at least one pair $n,m$ -- the result is an \emph{upper} bound for the coherence needed for incompatibility. The interpretation is that coherence between $|n\rangle$ and $|m\rangle$ cannot be precisely detected by a measurement capable of distinguishing these states. In particular, if $\coh_{nm}(\M)=1$ then $\M$ is incompatible with any $\F$ having $p_n^\F\neq p_m^\F$.

Necessary conditions for incompatibility require finding joint observables, equivalent to hidden variable models for quantum steering \cite{wolf09, uola15,quintino14,heinosaari15b,kiukas17, kiukas16}. This is hard to tackle analytically, and often restricted to single qubits or highly symmetric cases. Surprisingly, we now obtain a very general result using the Schur product theorem \cite{paulsen02}, which states that the entry-wise (Hadamard / Schur) product $A*B$ of PSD matrices $A$ and $B$ is also PSD. We call an observable $\P$ \emph{incoherent} if $\P(i)=\sum_{n=1}^d p^{\P}_n(i)|n\rangle\langle n|$ for all $i\in \Omega_\P$, and define a matrix $\inv(\P)$ by $\inv_{nm}(\P) =(1-d^2_{nm}(\P))^{-1}$ if $d_{nm}^2(\mathsf P)<1$ for all $n,m$.
\begin{proposition}\label{basicprop2} If $\P$ is incoherent and $\inv(\P) * \M(i)\geq 0$ for all $i\in \Omega_\M$, then $\M$ and $\P$ are jointly measurable.
\end{proposition}
\begin{proof} Define $C^\P(j)\geq 0$ by $c_{nm}^\P(j) =\sqrt{p_n^{\P}(j)p^{\P}_m(j)}$. Then $\G(i,j) := \inv(\P) *\M(i) * C^\P(j)\geq 0$ by the assumption and the Schur product theorem. But $\sum_j \G(i,j) = \M(i)$, and $\sum_i \G(i,j) = \inv(\P)* \id*C^\P(j)=\P(j)$ as $\P$ is incoherent. Hence $\G$ is a joint observable for $\M$ and $\P$.
\end{proof}
To appreciate how this result describes the coherence needed for incompatibility, note that the diagonal elements of $\mathsf P$ enter into the matrix $\mathfrak S(\mathsf P)$, while the positivity condition describes the (lack of) coherence in $\mathsf M$. More specifically, when the coherences ${\rm coh}_{nm}(\mathsf \M)$ are small enough relative to $1-d^2_{nm}(\mathsf P)$, then the off-diagonal elements $\langle n|\mathsf M(i)|m\rangle (1-d^2_{nm}(\mathsf P))^{-1}$ of the matrix $\mathfrak S(\mathsf P) *\mathsf M(i)$ are small relative to the unit diagonal, and hence (e.g. by the Sylvester determinant criterion), the positivity condition $\mathfrak S(\mathsf P) *\mathsf M(i)$ of Prop. \ref{basicprop2} will hold. This ensures the existence of a joint observable, showing that the (collective) coherence in $\mathsf M$ is \emph{not} enough for incompatibility. In examples with suitable parametrisation, this then translates into a \emph{lower bound} for the coherence needed for incompatibility.

\subsection{Basic examples} We now link the above results to the noise bounds for incompatibility \cite{designolle19a,designolle19b,chitambar19, bluhm20}: consider 
\begin{equation}\label{whitenoise}
\P_\alpha(j) =  \alpha |j\rangle\langle j| + (1-\alpha) d^{-1} \id, \quad 0\leq \alpha \leq 1.
\end{equation}
Let $\alpha_{\M}$ be the minimal $\alpha$ for which a given observable $\M$ is incompatible with $\P_\alpha$; this is a way of quantifying \emph{incompatibility-robustness} of $\P_1$ relative to $\M$ \cite{designolle19b}. Now $\P_\alpha$ has only one Hellinger distance; $d^2_{nm}(\P_\alpha) = 1-g_d(\alpha)$ and $\inv_{nm}(\P_\alpha) =1/g_d(\alpha)$ for $n\neq m$, with 
\begin{equation*}
g_d(\alpha) = \tfrac 1d \left((d-2)(1-\alpha) + 2\sqrt{1-\alpha}\sqrt{1+(d-1)\alpha}\right)\,.
\end{equation*}
Here $\alpha\mapsto d^2_{nm}(\P_\alpha)$ is monotone increasing, setting up a correspondence between $\alpha_\M$ and the Hellinger distance, the latter providing a link to coherence via Prop. \ref{basicprop} and Prop. \ref{basicprop2}:
\begin{equation*}\label{hellingerbound}
\max_{n, m\neq n} \coh_{nm}(\M)\leq g_d(\alpha_\M) \leq \max_{n,i} \sum_{m\neq n} \frac{|\langle n|\M(i)|m\rangle|}{p_n^\M(i)}.
\end{equation*}
The upper bound follows from Prop. \ref{basicprop2}, as $\inv(\P_\alpha) * \M(i)$ is \emph{diagonally dominant} \cite{horn12}, hence PSD, if $g_d(\alpha)$ exceeds this bound. As a simple example take a qubit with $\sigma_z$-basis as the incoherent basis. Then $\P_\alpha(0)=\frac 12(\id + \alpha \sigma_z)$, so \emph{any} observable $\M$ with $\coh_{01}(\M)>g_2(\alpha) = \sqrt{1-\alpha^2}$ is incompatible with $\P_\alpha$. If $\M$ is binary with $\M(0)=\frac 12(\id + {\bf m} \cdot \sigma)$, we have $$1\leq \frac{g_2(\alpha_\M)^2}{m_1^2+m_2^2}\leq\max \left\{\frac{1}{1+m_3},\frac{1}{1-m_3}\right\}.$$ We can check this using the standard qubit criterion \cite{busch86}, according to which $g_2(\alpha_\M)^2=(m_1^2+m_2^2)/(1-m_3^2)$; hence our bounds are exact iff $m_3=0$.

Next we obtain a quick proof for the known noise bound for the incompatibility of MUBs \cite{carmeli12,uola16,designolle19a,carmeli19a,carmeli19b}: 
\begin{proposition}\label{noisymubs} Let $\M(i) = \lambda \Q_0(i) +(1-\lambda) d^{-1}\id$ where $\Q_0$ is mutually unbiased to the incoherent basis. Then $\alpha_{\M}=g_d(\lambda)$ for any $\lambda\in [0,1]$.
\end{proposition}
\begin{proof}
The crucial observation is that $\M = C*\Q_0$, where $C$ has unit diagonal and $c_{nm} =\lambda$ for $n\neq m$. Hence $\coh_{nm}(\M) =\lambda$, so $\lambda \leq g_d(\alpha)$ by Prop.~\ref{basicprop} if $\M$ and $\P_\alpha$ are jointly measurable. Conversely, if $\lambda \leq g_d(\alpha)$ then $\inv(\P_\alpha) *C\geq 0$, so $\inv(\P_\alpha) * \M(i)= \inv(\P_\alpha) *C*\Q_0(i)\geq 0$, so $\P_\alpha$ and $\M$ are jointly measurable by Prop.~\ref{basicprop2}. Therefore $\alpha_{\M}=g_d(\lambda)$ as $g_d^{-1}=g_d$ on $[0,1]$.
\end{proof}
We note that $g_d(\lambda)$ is precisely the bound obtained in the cited literature by other methods. We will return to this example later.

\section{Incompatibility due to a coherence pattern}\label{cohmatrix}
Here we specialise to the physically relevant class of observables, introducing first their general structure, and then focusing on the spin-boson model.

\subsection{General consideration of coherence matrices}\label{gencoh}

Starting with a brief motivation, we consider again the above qubit example: when $m_3=0$, we have $\M = C * \Q_0$ where $C$ is a PSD matrix with $c_{00}=c_{11}=1$, $c_{01} =m_1-i m_2$, and $\Q_0(0)=\frac 12(\id +\sigma_x)$. Note that $\coh_{01}(\Q_0)=1$ (maximal coherence), and $\M$ is a ``noisy'' version of $\Q_0$ obtained via \emph{pure decoherence} \cite{buscemi05,helm09,breuer02}. The same structure appears in the case of MUBs (proof of Prop. \ref{noisymubs}), in any dimension. Accordingly, we call any PSD matrix $C$ with unit diagonal a \emph{coherence (pattern) matrix}, and consider noisy observables $\M(i) = C*\Q(i)$. In the special case where $\Q_0$ is maximally coherent, we have $\coh_{nm}(\M) =|c_{nm}|$ so the coherence in $\M$ is ``imprinted'' by $C$. This setting captures a remarkable interplay of maximal incompatibility and coherence. Indeed, in Appendix \ref{A:mainthm} we use dilation theory to prove the following result:
\begin{theorem}\label{mubthm}
Let $C$ be a coherence matrix and $\Q_0$ a maximally coherent observable. If an incoherent observable $\P$ is jointly measurable with $C*\Q_0$, it is jointly measurable with $C* \Q$ for \emph{every} observable $\Q$.
\end{theorem}
Hence, any maximally coherent observable (such as one mutually unbiased to the incoherent basis) is also maximally incompatible in this setting (which is not true in general \cite{designolle19a}). Incompatibility arising from coherence is now described as follows:
\begin{definition}\label{def1}
For a coherence matrix $C$, we denote by $\mathcal C_{C}$ the set of incoherent observables $\P$ jointly measurable with $C*\Q$ for every observable $\Q$. If $\P\notin \mathcal C_C$ we say that $\P$ \emph{has incompatibility due to the coherence (pattern) $C$}.
\end{definition}

Note that to find out whether a given $\P$ lies in $\mathcal C_C$, it suffices (by Thm.~\ref{mubthm}) to check whether $\P$ is jointly measurable with $C*\Q_0$ for some \emph{fixed} maximally coherent observable $\Q_0$. The general results of the preceding section have the following useful corollaries:
\begin{corollary}\label{cohcor}
Let $C$ be a coherence matrix. Any $\P\in \mathcal C_{C}$ has $d_{nm}^2(\P)\leq 1-|c_{nm}|$ for all  $n,m$.
\end{corollary}
\begin{proof} Follows by Prop. \ref{basicprop}, as $\coh_{nm}(C*\Q_0) = |c_{nm}|$.
\end{proof}
\begin{corollary}\label{cohcor2} Let $C$ be a coherence matrix and $\P$ an incoherent observable. 
If $C*\inv(\P)\geq 0$ then $\P\in \mathcal C_{C}$.
\end{corollary}
\begin{proof}
Let $\Q$ be an arbitrary observable, and $\M = C*\Q$. Then $\inv(\P)*\M(i) = (C*\inv(\P))*\Q(i)\geq 0$ by the Schur product theorem, so $\P$ is jointly measurable with $\M$ by Prop.~\ref{basicprop2}, hence $\P\in \mathcal C_C$.
\end{proof}
The coherence matrix model is strongly motivated by open quantum systems. In fact, quantum coherence is notoriously fragile against noise, and one of the basic mechanisms by which it decays is \emph{pure decoherence} (i.e.\ no dissipation), typically arising as subsystem dynamics from a unitary evolution on a larger system which leaves the incoherent basis unchanged \cite[Chapt. 4]{breuer02}. While incompatibility seems rarely tractable under general dynamics (see \cite{addis16} for a qubit case), our theory applies neatly to this type of dynamics. Each incoherent observable represents a conserved quantity, whose incompatibility with all other system observables is lost when the decaying coherence fails to sustain it; this characterises the emergent classicality of the open system in a more operational way than the decoherence itself.

More formally, suppose we have a family of coherence matrices $C[\lambda]$ depending on a parameter $\lambda\in [0,1]$. If $\lambda=\lambda(t)$ depends on a time parameter $0\leq t<\infty$, the map $\Lambda_t(\rho):=C[\lambda(t)]*\rho$ defines a quantum dynamical map, i.e. a family of completely positive trace preserving maps on the set of density matrices. A natural ``Markovianity'' property in this setting is 
\begin{equation}\label{divisibility}
C[\lambda\lambda']=C[\lambda]*C[\lambda'], \quad \lambda,\lambda'\in [0,1],
\end{equation}
which leads to the CP-divisibility of the dynamical map if the function $\lambda(t)$ is monotone decreasing. Then the loss of incompatibility is irreversible, and (as the Heisenberg picture evolution has the same form), the above corollaries can be used to bound the critical time at which a given incoherent observable loses its incompatibility by entering the set $\mathcal C_{\lambda(t)}$. Next we show, by considering a specific model, that this approach is amenable to analytical results even in large systems.

\subsection{Spin-boson model}\label{sb}

We consider the spin-boson model with collective interaction \cite{breuer02,palma96,unruh95} -- $N$ qubits coupled  to bosonic modes $b_k$ via the total spin $S_z = \frac 12\sum_{l=1}^N \sigma^{(l)}_z$. The total Hamiltonian is
$$H = H_S + \sum_k \omega_k b_k^\dagger b_k + \sum_k S_z (g_kb_k^\dagger +\overline g_k b_k),$$
where $H_S = \omega_0 S_z$ is the system Hamiltonian. The incoherent basis is the $\sigma_z$ basis $\{|{\bf m}\rangle\}$, where ${\bf m}=(m_1,\ldots,m_N)$, and we let $|{\bf m}|=\sum_l m_l$. 
With the bath initially in a thermal state, the system state at time $t$ is $\rho(t) = C[\lambda(t)] *\rho_0$, where $C[\lambda]$ is the coherence matrix $c_{{\bf n},{\bf m}}[\lambda]=\lambda^{(|{\bf n}|-|{\bf m}|)^2}$, and $\lambda(t)\in [0,1]$ is given by the bath temperature and spectral density \cite[Sec. 4.2]{breuer02}. In the Heisenberg picture, the dynamics transform the system observables $\Q$ into $C[\lambda(t)] * \Q$, which is precisely of the form considered above, and has the divisibility property \eqref{divisibility}. The task is to characterise the set $\mathcal C_{C[\lambda]}$ for $\lambda\in [0,1]$.

If $t\mapsto \lambda(t)$ is monotone decreasing (as in \cite[p. 230]{breuer02}), the loss of incompatibility is irreversible due to \eqref{divisibility}. However, the model also has \emph{decoherence-free subspaces} (DFS) $\mathcal D_{j}={\rm span}\{|{\bf n}\rangle\mid |{\bf n}| =j\}$ \cite{palma96,lidar14, lidar98, bacon00}; basis elements in the same DFS have $c_{{\bf n,m}}[\lambda]=1$ for all $\lambda$. By Cor. \ref{cohcor}, each DFS ``protects'' the incompatibility of any $\P$ not proportional to $\id$ inside it. Observables $\P$ exhibiting a transition to classicality therefore have $p_{\bf n}^\P =p_{\bf m}^\P$ when $|{\bf n}\rangle, |{\bf m}\rangle$ lie in the same DFS. In Appendix \ref{A:symmetry} we show that for these $\P$ the problem reduces to an $N+1$-dimensional space with incoherent basis $\{|k\rangle\}$ indexed by the DFS labels $k=0,\ldots, N$: we have $\P\in \mathcal C_{C[\lambda]}$ iff $\tilde \P\in \mathcal C_{\tilde C[\lambda]}$ where $p_k^{\tilde \P} = p_{\bf n}^\P$ for $|{\bf n}\rangle\in \mathcal D_k$, and the coherence matrix is
\begin{equation}\label{sbmatrix}
\tilde C[\lambda] =
\begin{pmatrix} 1 & \lambda & \lambda^4 &\cdots & \cdots &\lambda^{N^2}\\
 \lambda &1 & \lambda & \ddots & & \vdots\\
 \lambda^4 & \lambda & \ddots &\ddots &\ddots & \vdots\\
 \vdots & \ddots & \ddots & \ddots & \lambda & \lambda^4\\
 \vdots & & \ddots & \lambda & 1 &\lambda\\
 \lambda^{N^2} & \cdots & \cdots & \lambda^4 &\lambda & 1\end{pmatrix}.
\end{equation}
We further focus on the measurements of the DFS label $j=0,\dots,N$ which are \emph{covariant} for the permutation $j\mapsto N-j$ leaving $C[\lambda]$ invariant, i.e. $p_{\bf n}^\P(j)=p_{\bf m}^\P(N-j)$ when $|{\bf n}| = N-|{\bf m}|$. We denote by $\mathcal C_{C[\lambda]}^{\rm sym}$ the set of covariant $\P\in \mathcal C_{C[\lambda]}$. This set has affine dimension $\frac 12 N(N+1)$, and we find it analytically for $N=2$ in Appendix \ref{A:sbN2};  see Fig.\ \ref{fig:spinboson1}. In Section \ref{symmetry} and Appendix \ref{A:symmetry} we develop a general theory of covariance systems for coherence matrices.

\begin{figure}[t]
\begin{minipage}{\columnwidth}
\includegraphics[scale=0.24]{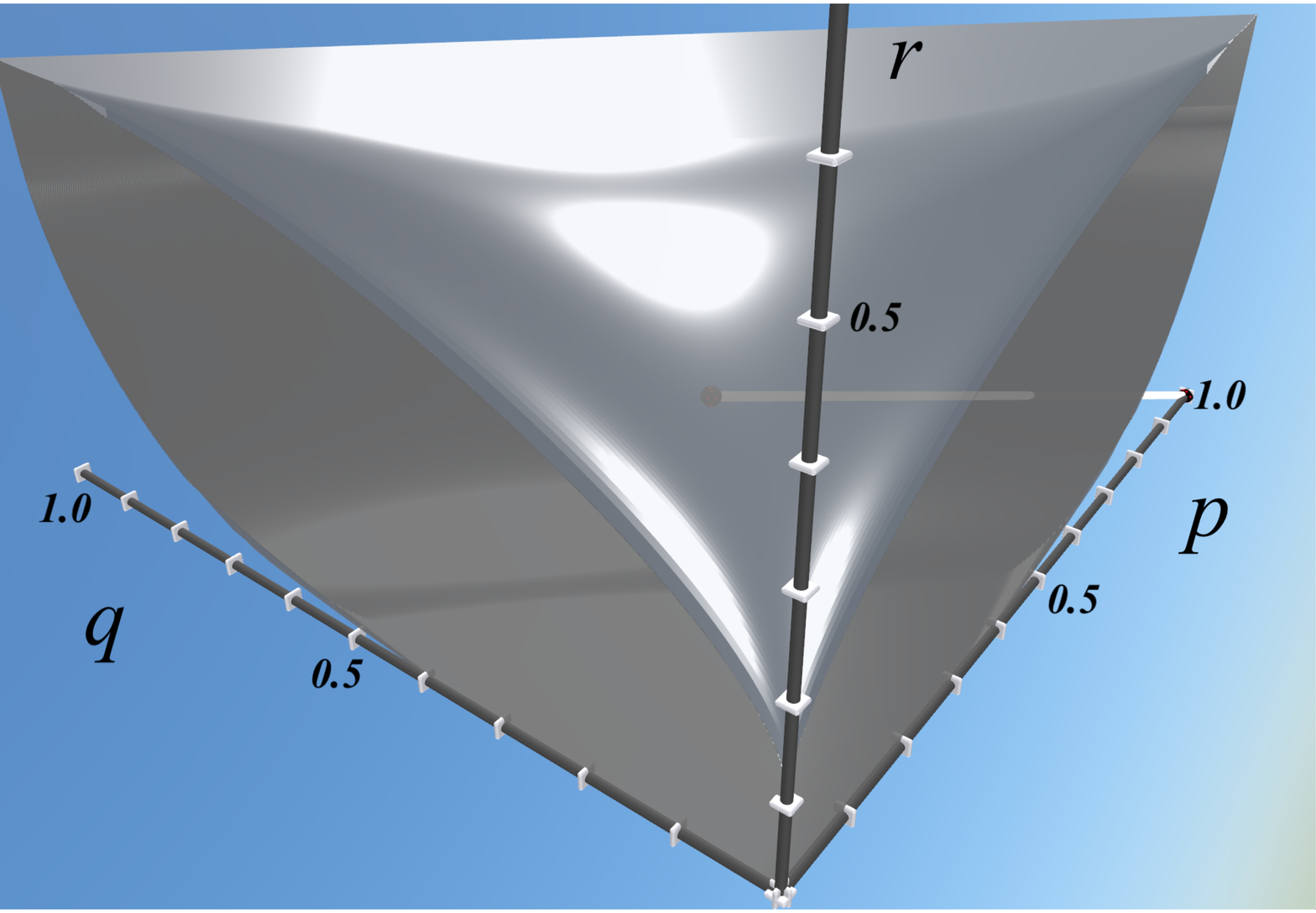}
\end{minipage}
\caption{\label{fig:spinboson1} (Color online) The convex set $\mathcal C^{\rm sym}_{C[\lambda]}$ for $N=2$, $\lambda =0.7$, parametrised by $(p,q,r)$ where $p=p^{\P}_{1}(0)$, $q=p^{\P}_4(0)$, $r=p^{\P}_2(0)=p^{\P}_3(0)$. The line $\{\P_\alpha^{C}\mid \alpha\in [0,1]\}$ is shown connecting $(\tfrac 14,\tfrac 14,\tfrac 14)$ to $(1,0,0)$; it enters the set at $\alpha =\alpha_2(0.7)\approx 0.58$.}
\end{figure}
\begin{figure}[t]
\begin{minipage}{0.465\columnwidth}
\begin{flushleft}
\includegraphics[scale=0.34]{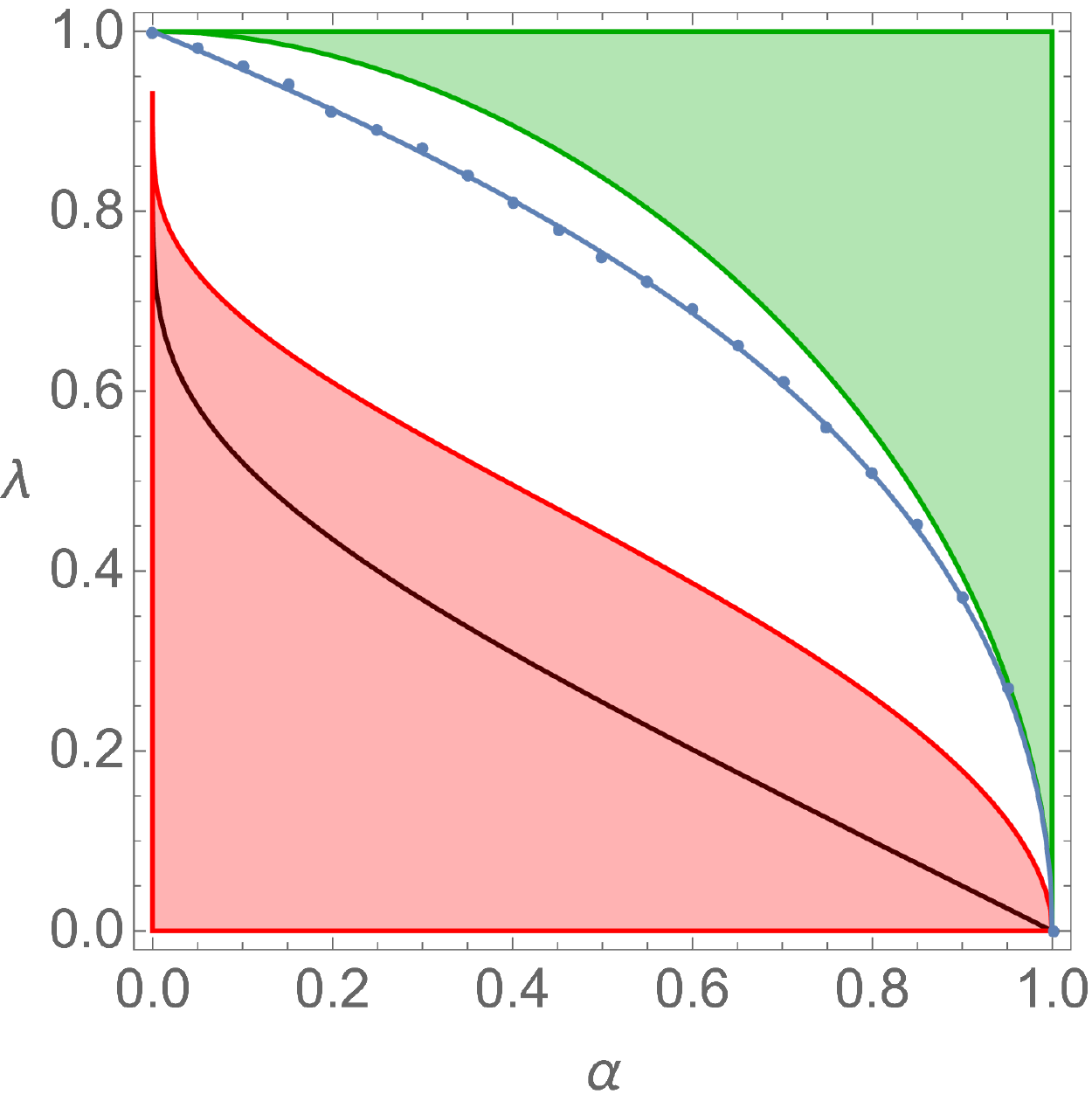}
\end{flushleft}
\end{minipage}
\begin{minipage}{0.52\columnwidth}
\begin{flushleft}
\includegraphics[scale=0.34]{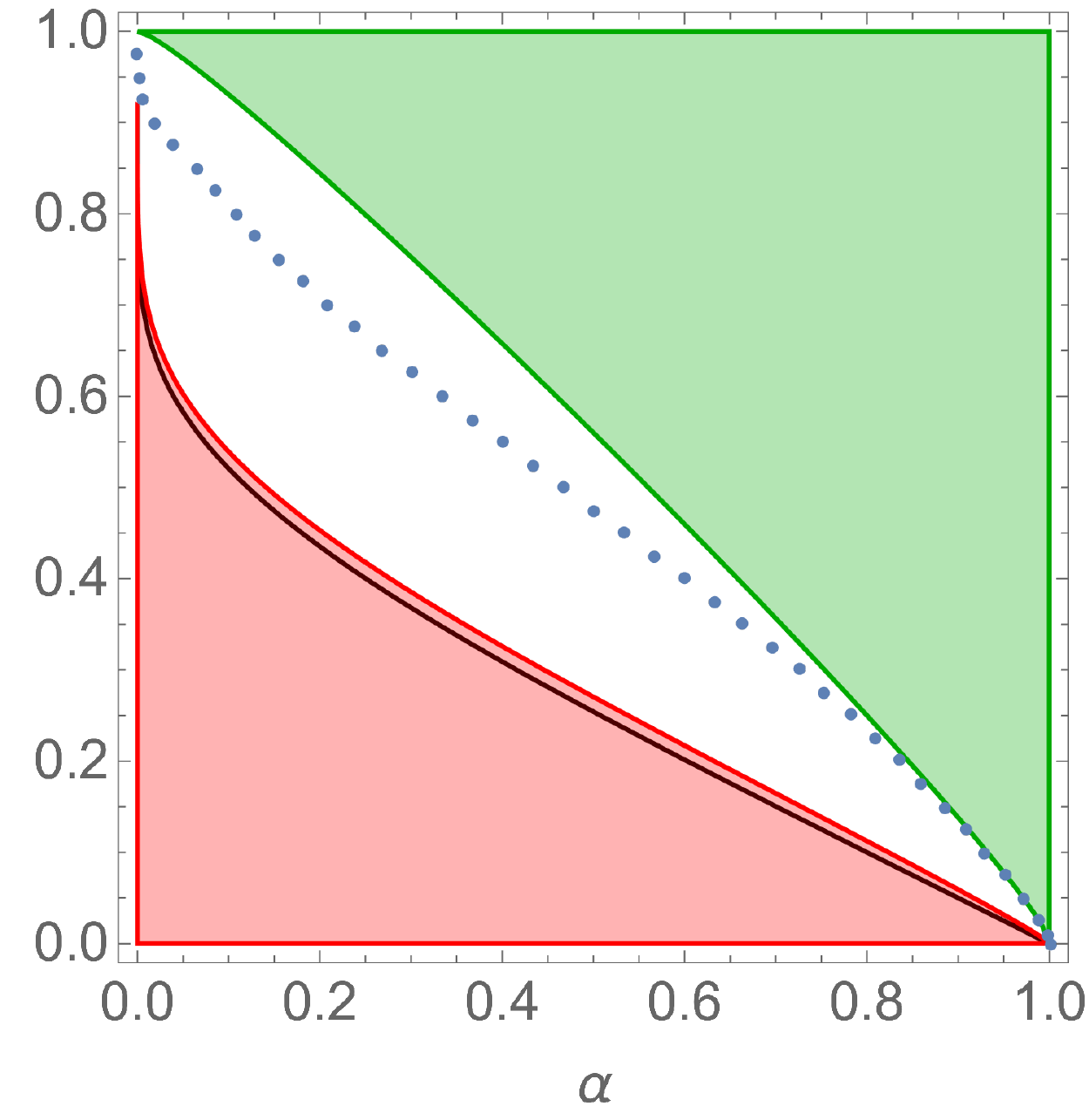}
\end{flushleft}
\end{minipage}
\caption{\label{fig:spinboson2} (Color online). Incompatibility of the incoherent observables arising as $\alpha$-weighted averages of the system Hamiltonian and its depolarisation, in the spin-boson model with coherence parameter $\lambda$, for $N=2$ (left) and $N=10$ (right). The green (resp. red) region has incompatibility (loss thereof) ``certified'' by the \emph{analytical} bound $U_N$ ($L_N$) of Prop.\ \ref{sbbounds}. Below the $\vartheta_3(\frac \pi 2, \lambda)$ bound (black line) incompatibility is lost for \emph{all $N$}. The true boundary $\alpha_N$ is given for comparison by a semidefinite program (SDP) (blue dots), and analytically for $N=2$ (solid blue line).}
\end{figure}

For larger $N$ we focus on covariant observables $\P^C_{\alpha} := \alpha \P^C +(1-\alpha)\P_{\rm dep}^C$, $\alpha\in [0,1]$, where $\P^C$ is the spectral resolution of $H_S$ (so $\P^C(j)$ is the projector onto the DFS $\mathcal D_j$), and $\P^C_{\rm dep}$ its depolarisation $\P_{{\rm dep}}^C(j):=\frac{1}{2^{N}}{\rm tr}[\P^C(j)]\id=c_{N}(j)\id$ into the ``coin toss'' distribution $c_{N}(j)= \binom Nj \tfrac{1}{2^{N}}$. The proportion of $\{\P^C_\alpha \mid \alpha\in [0,1]\}$ having lost incompatibility at time $t$ is $\alpha_N(\lambda(t))$ where $$\alpha_N(\lambda) :=\max \{\alpha>0 \mid \P^C_{\alpha}\in \mathcal C_{C[\lambda]}\}.$$ With the reduction \eqref{sbmatrix} we find
$\alpha_N(\lambda)=\max\{\alpha>0 \mid  \tilde \P_\alpha^C \in \mathcal C_{\tilde C[\lambda]}\}$, where $\tilde \P_\alpha^C(j) = \alpha|j\rangle\langle j| + (1-\alpha) c_N(j) \id$. The task is to find analytical bounds for $\alpha_N(\lambda)$. For $N=2$ we can explicitly solve (Appendix \ref{A:sbN2}) $\alpha_2(\lambda) = 1-\tfrac{4\lambda^2}{3+\lambda^4+2\sqrt2(1-\lambda^2)}$. Crucially, the Hellinger distances are tractable for \emph{any} $N$: $d_{kk'}(\tilde \P^{C}_\alpha)=1-\beta_{kk'}(\alpha)$ when $k\neq k'$, with $\beta_{kk'}(\alpha)=1-\alpha + u_{k}+u_{k'}$, $u_{k} =\sqrt{c_N(k)(1-\alpha)\big(\alpha + c_N(k)(1-\alpha) \big) } -c_N(k)(1-\alpha)$. Each $\alpha\mapsto \beta_{kk'}(\alpha)$ is decreasing in $\alpha$; using the inverse functions $[\beta_{01}]^{-1}$, $[\beta_{00}]^{-1}$ we set $U_N(\lambda)=[\beta_{01}]^{-1}(\lambda)$ and $L_N(\lambda) = [\beta_{00}]^{-1}(1-\vartheta_3(\tfrac \pi 2,\lambda))$, where $\vartheta_3(x,\lambda) = 1+2\sum_{k=1}^\infty \lambda^{k^2}\cos(2kx)$ is the Jacobi Theta function \cite{montgomery88}. The following result holds:

\begin{proposition}\label{sbbounds}  $\vartheta_3(\frac \pi 2, \lambda)\leq L_N(\lambda) \leq \alpha_N(\lambda) \leq U_N(\lambda)$ for all $N=1,2,\ldots$, and $\lambda\in [0,1]$.
\end{proposition}

\begin{proof} We fix $\lambda\in [0,1]$ (and hence also $\tilde C[\lambda]$, $U_N=U_N(\lambda)$, $L_N=L_N(\lambda)$, $\alpha_N=\alpha_N(\lambda)$). Now if $\alpha$ is such that $\tilde \P^{C}_\alpha\in \mathcal C_{\tilde C}$, then $\lambda \leq \beta_{01}(\alpha)$ by Cor.\ \ref{cohcor}, i.e.\ $U_N\geq \alpha$, so $U_N\geq \alpha_N$; this establishes the upper bound. For the lower bounds, define $B=(b_{nm})$ by $b_{nm} = 1/\beta_{nm}$ for each $n,m$. (Notice that $\beta_{nn}\neq 1$, as $1-\beta_{nm}$ only coincides with the corresponding Hellinger distance on the off-diagonal elements.) Now $B\geq 0$, as $b_{nm} = \int_0^1 x^{\beta_{nm}-1} dx=\int_0^1 x^{-\alpha} x^{u_n} x^{u_m}dx$. Now $b_{nm}=(1-d^2_{nm}(\tilde \P^{C}_\alpha))^{-1}$ when $n\neq m$, but $b_{nn} = 1/\beta_{nn}\geq 1$, so $\tilde C*\mathfrak S(\tilde \P_\alpha^C) = \tilde C*B +D$ with $D$ diagonal, $d_{nn}= 1-b_{nn}$. Since $B\geq0$, we get $(\tilde C-r\id)*B\geq 0$, where $r$ is the bottom eigenvalue of $\tilde C$. By the theory of Toeplitz matrices (\cite[p. 194, 211]{gray06}, \cite[Lemma 1]{montgomery88}), $$r\geq \min_{x\in [0,2\pi]} \sum_{k=-\infty}^\infty \tilde c_{k0}\,e^{ikx}=\min_{x\in [0,2\pi]}\vartheta_3(\frac x2,\lambda) = \vartheta_3(\frac \pi 2,\lambda).$$ Hence, if $\vartheta_3(\frac\pi 2,\lambda)\geq 1-\beta_{00}(\alpha)$ then $\tilde C*\mathfrak S(\tilde \P_\alpha^C) \geq r (\id -D)+D\geq (1-\beta_{00})(\id -D)+D\geq 0$ as $1-\beta_{00}=\max_{n}(1-\beta_{nn})$, and so $\tilde \P_\alpha^{C}\in \mathcal C_{\tilde C}$ by Cor.\ \ref{cohcor2}. Hence $\tilde \P_\alpha^{C}\in \mathcal C_{\tilde C}$ for all $\alpha\leq L_N$, so $\alpha_{N} \geq L_N$. Finally, $\vartheta_3(\frac \pi 2, \lambda)\leq L_N$ since $\alpha\geq 1-\beta_{00}(\alpha)$ for all $\alpha$. This completes the proof.\end{proof}

Fig.~\ref{fig:spinboson2} shows the bounds for $N=2,10$, and the analytical curve $\alpha_2(\lambda)$ with a numerical consistency check (blue dots) computed with a generic joint measurability SDP available in \cite{cavalcanti16} applied to $(\tilde C[\lambda] *\Q_0,\tilde\P^{C}_\alpha)$ (see Thm.\ \ref{mubthm}). This SDP is not practical for large $N$; to compute $\alpha_{10}(\lambda)$ we used the efficient SDP \eqref{GII0} adapted to our setting as described in the next section (implemented in Python \cite{cvxopt}). Prop.\ \ref{sbbounds} says that at any time $t$, at least the proportion $1-U_N(\lambda(t))$ of the line $\{\P^C_\alpha\mid \alpha\in [0,1]\}$ has incompatibility due to coherence, while at least the proportion $L_N(\lambda(t))\geq \vartheta_3(\frac \pi 2, \lambda(t))$ has lost it. Remarkably, the last bound is independent of $N$, i.e. holds for any system size. Finally, the bound $U_N$ is tight near the classical limit (small $\lambda$): 
\begin{proposition}\label{asymptotic} For each fixed $N$, the curves $\alpha = \alpha_N(\lambda)$ and $\alpha = U_N(\lambda)$ have the same asymptotic form:
$\lambda = \tfrac{1+\sqrt{N}}{2^{N/2}}\sqrt{1-\alpha} + O(1-\alpha)$  as $\alpha\rightarrow 1$ ($\lambda\rightarrow 0$).
\end{proposition}
We prove this (Appendix \ref{A:smallcoh}) by explicitly constructing relevant joint observables in the operational framework of the next section. Note that the square root behaviour near the classical limit is distinct from the ``middle'' regime for $\lambda$, where $U_N(\lambda)$ decreases towards $1-\lambda$ as $N$ increases. Clearly, incompatibility is much more intricate than the coherence which sustains it in this model.

\section{Operational framework for coherence and incompatibility}\label{operational}
Having demonstrated our theory in applications, we now gain further insight by reformulating it in resource-theoretic terms, in the context of quantum measurement theory.

\subsection{Resource-theoretic aspects}

Here we describe how our theory formally integrates into the resource theory of quantum coherence, and especially measurement coherence. We stress that our aim is \emph{not} to develop a comprehensive joint resource theory for coherence and incompatibility, but rather to focus on the most relevant aspects.

The ``free resources'' are the incoherent observables $\P(i) = \sum_n p^{\P}_n(i) |n\rangle\langle n|$ \cite{oszmaniec19, baek20} already used above; they are jointly measurable with each other. Their non-classicality is quantified by the Hellinger distances: $d^2_{nm}(\P_1) =1$ for all $n\neq m$, while $d_{nm}(\P)=0$ for all $n,m$ iff each $\P(j)$ is a multiple of the identity, i.e. $\P$ is jointly measurable with \emph{every} observable. An observable $\P$ is incoherent iff $\coh_{nm}(\P) =0$ for all $n\neq m$.

Any quantum channel, a completely positive (CP) trace preserving map $\Lambda$, is a ``free operation'' for incompatibility \cite{heinosaari15a, heinosaari15b,chitambar19}, acting on observables via preprocessing $\Lambda^*(\M)(i) = \Lambda^*(\M(i))$, where $\Lambda^*$ is defined by $\tr{\Lambda^*(X)\rho}=\tr{X\Lambda(\rho)}$ for all matrices $X$, $\rho$. If $\Lambda(|n\rangle\langle n|)=|n\rangle\langle n|$ for all $n$, i.e. $\Lambda$ leaves each incoherent observable unchanged, then $\Lambda$ is also ``free'' for coherence, called a \emph{genuinely incoherent operation (GIO)} \cite{devicente17,yao17}. These have already appeared in our setting: each GIO has the form $\Lambda^*_C(\M)(i) = C*\M(i)$ for some coherence matrix $C$ \cite{buscemi05, devicente17}. The entry-wise coherence is monotonic in GIOs, as $\coh_{nm}(\Lambda^*_C(\M)) = |c_{nm}| \coh_{nm}(\M)\leq \coh_{nm}(\M)$.

Observables are measured by \emph{instruments} $\I=(\I_i)_{i\in \Omega}$, where each $\I_i$ is a CP map such that $\sum_i \I_i$ is a channel \cite{QM}. 
The observable measured by $\I$ is $\M(i)=\I^*_i(\id)$. For a state (density matrix) $\rho$ the post-measurement state given outcome $i$ is $\I_i(\rho)$. Instruments are needed for sequential implementation of joint measurements: measuring first $\M$ with $\I$, and then an observable $\F$, we get a joint POVM $\G(i,j) =\I^*_i(\F(j))$ for $\M$ and $\Lambda^*(\F)$. Joint measurements usually do not have sequential implementations (unless one ``cheats'' by allowing a larger output space \cite[Prop.\ 2]{Optimal}). We define a \emph{genuinely incoherent instrument} (GII) as one whose channel $\sum_i\I_i$ is a GIO. It follows (Appendix \ref{A:dilations}) that any GII with channel $\Lambda_C$ has the form $\I^*_i(X) = C(i)* X$ for some PSD matrices $C(i)$ with $\sum_i C(i) = C$; we call $C$ the coherence matrix of $\I$. GIIs are free operations also for coherence: they cannot create coherent observables from incoherent ones by sequential combination. The observable measured by a GII is incoherent, namely
$\P(i) := \I^*_i(\id)=\sum_n c_{nn}(i) |n\rangle\langle n|$. 

\subsection{Sequential measurement setting}

Given the above concepts, the operational scheme in Fig.\ \ref{fig:scheme} naturally emerges; in this setting, coherence needed for incompatibility can now be characterised as follows:
\begin{figure}[t]
\begin{minipage}{\columnwidth}
\includegraphics[scale=0.29]{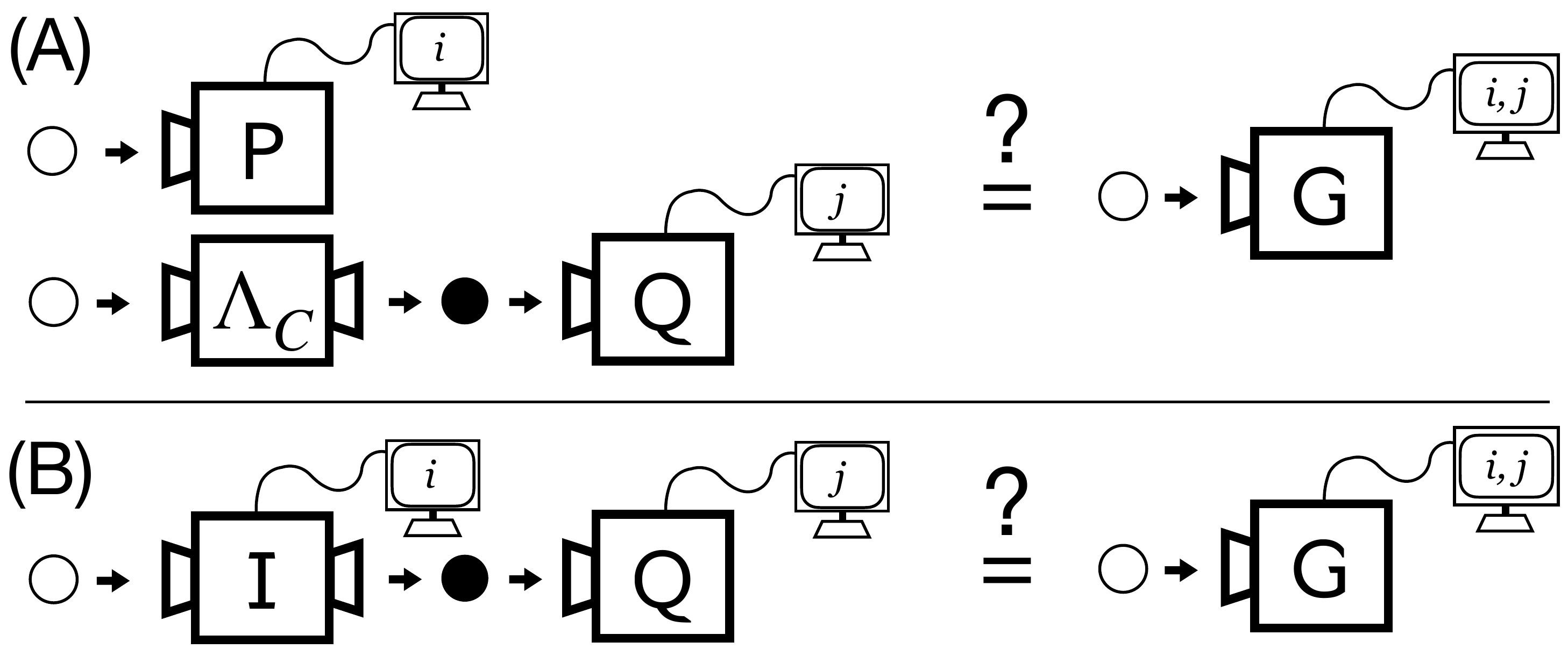}
\end{minipage}
\caption{\label{fig:scheme} (A): Operational setting for linking incompatibility and coherence. Here $\P$ is an incoherent observable and $\Lambda_C$ a channel which imprints a coherence pattern $C=(c_{nm})$ into a coherent observable $\Q$ by Hadamard multiplication. If $C*\Q$ is jointly measurable with $\P$ (via a joint POVM $\G$) for \emph{every} $\Q$,
the coherence $C$ cannot sustain any (pairwise) incompatibility involving $\P$, and
we write $\P\in \mathcal C_{C}$. (B): Sequential implementation of the joint POVM $\G$ in (A), given a GII $\I$ with observable $\P$ and channel $\Lambda_C$. By Thms.\ \ref{mubthm} and \ref{mainthm}, the existence of $\I$ is equivalent to joint measurability in (A) when $\Q$ is a maximally coherent.}
\end{figure}
\begin{theorem}\label{mainthm} Let $C$ be a coherence matrix and $\P$ an incoherent observable. The following are equivalent:
\begin{enumerate}
\item[(i)] $\P\in \mathcal C_{C}$;
\item[(ii)] There exists a GII with coherence matrix $C$ and observable $\P$, that is, matrices $C(j)$ satisfying
\begin{align}\label{GII0}
C(j)&\geq 0, & \textstyle\sum_{j} C(j) &= C, & c_{nn}(j) = p^{\P}_n(j).
\end{align}
\end{enumerate}
In that case a joint measurement of $\P$ and $C*\Q$, for any observable $\Q$, can be implemented sequentially by first measuring $\P$ using the GII in (ii), and subsequently $\Q$.
\end{theorem}

The crucial part of the proof of this result is the construction of the special joint observable given in the proof of Thm. \ref{mubthm} (see Appendix \ref{A:mainthm}). The result is surprising, as joint observables with a GII implementation are quite special, requiring channel-observable compatibility \cite{heinosaari13,heinosaari14,heinosaari18}. We note that \eqref{GII0}, as a SDP \cite{boyd04}, is more efficient than a generic joint measurability SDP due to lower dimensionality, but still not analytically solvable except in simple cases; see Appendix \ref{A:centrod3} for a $d=3$ example. However, there is a useful class of GIIs: for each incoherent observable $\P$ define $C^{\P}(j)$ as in the proof of Prop.~\ref{basicprop2}; this is a GII with coherence matrix $C^{\P}$ given by the Hellinger distances: $c^{\P}_{nm} = \sum_j c^{\P}_{nm}(j)= 1-d^2_{nm}(\P)$. Following the proof of Prop.~\ref{basicprop2} we obtain a GII for the setting in Cor. \ref{cohcor2}:
\begin{proposition}\label{GIIprop} If $C*\inv(\P)\geq 0$ then $C(j) = C*\inv(\P) * C^{\P}(j)$ defines a GII with coherence matrix $C$ and observable $\P$. 
\end{proposition}
\begin{proof}
Now $\sum_j C(j) = C*\inv(\P)* \sum_j C^{\P}(j) = C$, and $c_{nn}(j) = c^{\P}_{nn}(j) = p^{\P}_n(j)$. If $C*\inv(\P)\geq 0$ then $C(j)\geq 0$ by the Schur product theorem, so \eqref{GII0} holds for $C(j)$.
\end{proof}
As an example of a GII we consider the qubit case:
\begin{align*}
C &=\begin{pmatrix} 1 & c\\ \overline{c}& 1\end{pmatrix}, & \P(j) &= \begin{pmatrix}p_j & 0\\ 0 & q_j\end{pmatrix}, \quad j\in \mathbb N,
\end{align*}
with $c\in \mathbb C$, $|c|\leq 1$, and $(q_j)$, $(p_j)$ being probability distributions. The Hellinger distance is $d_{01}^2(\P) = 1-\gamma$ with $\gamma=\sum_j \sqrt{p_jq_j}$, and we obtain
$\P\in \mathcal C_{C}$ iff $|c| \leq \gamma$. Indeed, $\P\in \mathcal C_{C}$ implies $|c| \leq \gamma$ by Cor.\ \ref{cohcor}, while $|c| \leq \gamma$ implies $\P\in \mathcal C_{C}$ by Cor.\ \ref{cohcor2}, with the GII
$$
C(j) =\begin{pmatrix}p_j & c \gamma^{-1} \sqrt{q_jp_j} \\
\overline c\gamma^{-1} \sqrt{q_jp_j} & q_j\end{pmatrix}.
$$

Revisiting the noisy MUB example (see Prop. \ref{noisymubs}), we note that any $-(d-1)^{-1}\leq \lambda \leq 1$ defines a valid coherence matrix $C$, and $\P_\alpha$ is an incoherent observable for the same range of $\alpha$. When $\lambda<0$ the bounds of Cor.\ \ref{cohcor} and \ref{cohcor2} do not coincide, but our general method applies (Appendix \ref{A:uniform}), reproducing the result in \cite{carmeli19b}, and (additionally) yielding a sequential GII implementation for all jointly measurable cases, including the ``corner'' $\lambda=\alpha = -(d-1)^{-1}$ where the L\"uders instrument fails. We consider this interesting exceptional case here, and postpone the rest of the proof to the Appendix. We have $\P_\alpha(j)=(d-1)^{-1}(\id -|j\rangle\langle j|)$, and $$C = d(d-1)^{-1}(\id -|\phi_d\rangle\langle \phi_d|),$$ where $\phi_d = d^{-\frac 12}\sum_k |k\rangle$. In this case the two corollaries do not tell us anything: the tradeoff in Cor.\ \ref{cohcor} is not violated, and the matrix in Cor.\ \ref{cohcor2} is not positive semidefinite. However, now $\P_{\alpha}\in \mathcal C_{C}$ directly by Thm.\ \ref{mainthm}, as we can construct a GII $C(j)$ satisfying \eqref{GII0}: $$C(j) := (d-2)^{-1} (\id - |\phi_d\rangle\langle \phi_d| - |\varphi_d^j\rangle\langle \varphi_d^j|),$$ where $\varphi^j_d := \sqrt{d(d-1)^{-1}}(|j\rangle -d^{-\frac 12}\phi_d)$. Indeed, positivity follows from the fact that $\varphi^j_d$ is orthogonal to $\phi_d$ for each $j$ (so each $C(j)$ is a multiple of a projection), and it is easy to check that $\sum_j C(j) = C$ and $c_{nn}(j) = (d-1)^{-1} (1-\delta_{jn}) = p^{\P_{\alpha}}_n(j)$.

In addition to these examples, in Appendix \ref{A:smallcoh} we explicitly construct a GII for the spin-boson model: it establishes the tight bound for the coherence needed for incompatibility near the classical limit, and proves Prop.~\ref{asymptotic}.

\subsection{Reduction by symmetry}\label{symmetry}
One of the main obstacles in our joint measurability problem is the difficulty of finding the form of a suitable GII. Symmetries in the coherence pattern $C$ can be used to simplify the search, and also single out relevant incoherent observables. We outline this reduction here; detailed derivations are given in Appendix \ref{A:symmetry}.

First note that the matrix $D$ of a unitary GIO $\Lambda_D(\rho) =U^{\dagger}\rho U$ has $d_{nm}=u_n\overline{u_m}$ with $|u_n|=1$ for all $n$; then $\Lambda_D$ changes neither incompatibility nor coherence. For coherence matrices $C,C'$, we write $C\simeq C'$ if $C=D *C'$ for a unitary GIO $\Lambda_D$. Clearly, $\mathcal C_C=\mathcal C_{C'}$.

Second, we call an incoherent $\P$ \emph{adapted to $C$} if $p^{\P}_n=p^{\P}_m$ ($\P$ does not distinguish $n$ from $m$) whenever $(n,m)$ has maximal coherence $|c_{nm}|=1$. Crucially, each $\P\in \mathcal C_C$ is adapted to $C$, as $|c_{nm}|=1$ implies $d^2_{nm}(\P)=0$ by Cor.\ \ref{cohcor}. So every $\P$ \emph{not} adapted to $C$ has incompatibility due to coherence. Now $\{1,\ldots, d\}$ splits into $n_C\leq d$ disjoint equivalence classes $I^C_k$, $k\in \Omega_C:=\{1,\ldots, n_C\}$, such that $I^C_k=\{m\mid |c_{nm}|=1\}$ for any $n\in I_k^C$. Let $\tilde{\mathcal H}$ be the Hilbert space with incoherent basis $\{|k\rangle\}_{k\in \Omega_C}$ and define $L:\tilde{\mathcal H}\to \mathcal H$ by $L|k\rangle = \sum_{n\in I^C_k} |n\rangle$. Then $C\simeq L\tilde CL^\dagger$ for a ``reduced'' coherence matrix $\tilde C$ acting on $\tilde{\mathcal H}$.

As an example, consider the spin-boson model: we have the coherence matrix $C=C[\lambda]$ with $n_C=N+1$ equivalence classes $I^C_{k} = \{{\bf n} \mid \sum_l n_l =k\}$, $k\in \Omega_C=\{0,1,\ldots, N\}$ of size $|I^C_k|=\binom{N}{k}$, corresponding to the decoherence-free subspaces ${\rm span}\{|{\bf n}\rangle\mid {\bf n}\in I^C_k\}$. The resulting reduced coherence matrix $\tilde C$ is easily seen to be the one given by \eqref{sbmatrix}.

Next, let $S_d$, the group of permutations of $\{1,\ldots,d\}$, act on $\mathcal H$ via $U_\pi|n\rangle=|\pi(n)\rangle$. We define
$$G_C:= \{\pi\in S_d\mid U_\pi^\dagger  C U_\pi \simeq C\}.$$
In Appendix \ref{A:symmetry} we show that $G_C$ is a permutation group, i.e. a subgroup of $S_d$, and we call it the \emph{symmetry group} of $C$. It moves each class $I_k^C$ as a whole, and hence gives rise to a map $\phi:G_C\to G_{\tilde C}$ through $\pi(I^C_k) = I^C_{\phi(\pi)(k)}$. Now let $G$ be any subgroup of $G_C$. We then say that an incoherent observable $\P$ is $G$-\emph{covariant} if it has outcome set $\Omega_C$, and $p^{\P}_{\pi^{-1}(n)}(j) = p^{\P}_n(\phi(\pi)(j))$ for $\pi\in G$, $j\in \Omega_C$. These observables have their outcomes directly linked to the equivalence classes of the basis labels. The joint measurability problem reduces considerably when restricted to them; we set $$\mathcal C_C[G]:=\{\P\in \mathcal C_C\mid \P \text{ is $G$-covariant}\}.$$ The case of the full symmetry group is denoted by $\mathcal C_C^{\rm sym}:=\mathcal C_C[G_C]$. In Appendix \ref{A:symmetry} we show that for $\P\in \mathcal C_C^{\rm sym}$, the SDP \eqref{GII0} can be constrained by a corresponding covariance condition at the GII level without any loss. Furthermore, we also link the incoherent observables on $\tilde{\mathcal H}$ and $\mathcal H$ as $\mathcal L(|k\rangle\langle k|) = \sum_{n\in I^C_k} |n\rangle\langle n|$; the reduction by symmetry is then given by
\begin{equation}\label{symreduction}
\mathcal C_C^{\rm sym}=\{\mathcal L(\P(\cdot))\mid \P\in \mathcal C_{\tilde{C}}[\phi(G_C)]\}.
\end{equation}

We note that $\phi(G_C)$ may be different from $G_{\tilde C}$. However, they coincide when $|I^{C}_k|=|I^{C}_{\pi(k)}|$ for each $\pi\in G_{\tilde C}$ and $k\in \Omega_C$, i.e. the equivalence classes linked by permutations in the reduced symmetry group have equal size. In that case we have the straightforward reduction $\mathcal C_C^{\rm sym}=\{\mathcal L(\P(\cdot))\mid \P\in \mathcal C_{\tilde{C}}^{\rm sym}\}$. An example is provided by the spin-boson model, where the reduced coherence matrix $\tilde C=\tilde C[\lambda]$ is invariant under the \emph{exchange permutation} $\pi_0$ defined by $\pi_0(k) = N-k$ for $k\in \Omega_{C}$. In fact, $G_{\tilde C}=\{e, \pi_0\}$ for all $\lambda$, and $|I^{C}_k|=|I^{C}_{\pi_0(k)}|$ for each $k$, so $\phi(G_C)=\{e, \pi_0\}$. Hence, by \eqref{symreduction}, the set $\mathcal C_{C}^{\rm sym}$ is isomorphic to $\mathcal C_{\tilde C}^{\rm sym}$, as stated in Section~\ref{sb}.

Finally, we give a simple example involving also complex phase factors, and demonstrating the case $\phi(G_C)\neq G_{\tilde C}$: Let $\lambda\in [0,1)$, and consider the following:
\begin{align*}
C &:=\begin{pmatrix} 1 & i & \lambda i\\
-i & 1 & \lambda \\
-\lambda i& \lambda & 1
\end{pmatrix}\simeq \begin{pmatrix} 1 & 1 & \lambda\\
1 & 1 & \lambda \\
\lambda  & \lambda  & 1
\end{pmatrix}  &\longrightarrow  &&\tilde C &= \begin{pmatrix} 1 & \lambda \\ \lambda & 1\end{pmatrix}.
\end{align*}
Here $d=3$, $n_C=2$, $\Omega_C=\{1,2\}$, $I^C_1=\{1,2\}$, $I^C_2=\{3\}$, $G_C=\{ e, (12)\}$, $G_{\tilde C}=S_2$, $\phi(G_C)=\{e\}$. Note that the symmetry of $C$ is ``revealed'' after factoring out a unitary GIO in the first step. As a further subtlety, the exchange symmetry of $\tilde C$ is excluded as $|I^C_1|\neq |I^C_2|$. The point of the reduction is that we can use the simpler two-dimensional case to solve the original three-dimensional joint measurability problem. Indeed, we first characterise $\mathcal C_{\tilde{C}}[\{e\}]$ as described in the qubit case after Prop. \ref{GIIprop} above, and then use \eqref{symreduction}: $\mathcal C_{C}^{\rm sym}$ consists of the binary observables $\P$ of the form
\begin{align*}
\P(1)&=\begin{pmatrix} p & 0 & 0\\ 0& p& 0\\0 & 0& q\end{pmatrix},&\P(2)&= \id -\P(1),
\end{align*}
where $p,q\in [0,1]$ and $\sqrt{pq}+\sqrt{(1-p)(1-q)}\geq \lambda$.

\section{Conclusion} We considered a general operational setting where quantum coherence is tightly linked to measurement incompatibility. We derived two explicit conditions for the coherence needed for incompatibility, and demonstrated that these are amenable to analytical calculations even in large open quantum systems. Topics of further study include the infinite-dimensional case and adaptation to quantum steering.

\section*{Acknowledgements} D.M. has received funding from the European Union's Horizon 2020 research and innovation programme under the Marie Sk\l{}odowska-Curie Grant Agreement No. 663830. D.M. also acknowledges financial support by the TEAM-NET project co-financed by the EU within the Smart Growth
Operational Programme (Contract No. POIR.04.04.00-00-17C1/18-00). J.K. thanks Alessandro Toigo and Teiko Heinosaari for a useful discussion on the noisy MUB problem.

\appendix

\section{Properties of the entry-wise coherence}\label{A:coh}

Recall that an observable is a POVM $\M$ on a Hilbert space $\mathcal H=\mathbb C^d$ with (finite) outcome set $\Omega_\M$, i.e. $\M(i)\geq 0$ for each $i\in \Omega_\M$, and $\sum_{i\in \Omega_\M} \M(i) =\id$. We assume that the outcome set is taken minimal, i.e. $\M(i)\neq 0$ for each $i\in \Omega_\M$. We fix a basis $\{|n\rangle\mid n=1,\ldots,d\}$ and call it the \emph{incoherent basis}. For any observable we denote $p_n^{\M}(i):= \langle n|\M(i)|n\rangle$ for all $n\in \{1,\ldots, d\}$, $i\in \Omega_\M$.

An observable $\M$ is \emph{incoherent} if each POVM element $\M(i)$ is diagonal, i.e. $\M(i) = \sum_{n=1}^d p^{\M}(i) |n\rangle\langle n|$ for all $i\in\Omega_\M$. It is \emph{mutually unbiased to the incoherent basis} if $\M(i) = |\psi_i\rangle\langle \psi_i|$ where $\{\psi_i\mid i\in \Omega_\M\}$ is a basis of $\mathcal H$ such that $|\langle n|\psi_i\rangle|^2=d^{-1}$ for each $n\in \{1,\ldots,d\}$ and $i\in \Omega_\M$. Note that incoherent observables can have arbitrary outcome set $\Omega_\M$, while $|\Omega_\M|=d$ for any if $\M$ is mutually unbiased to the incoherent basis.

In the main text we introduced the \emph{entry-wise coherence} and \emph{Hellinger distances} for each $n,m\in \{1,\ldots,n\}$:
\begin{align*}
\coh_{nm}(\M) &= \sum_{i\in \Omega_\M} |\langle n|\M(i)|m\rangle|,\\ d^2_{nm}(\M) &= 1-\sum_{i\in \Omega_\M} \sqrt{p_n^{\M}(i)p_m^{\M}(i)}.
\end{align*}
We are not aware of the entry-wise coherence having appeared in the literature as such, but it has been used recently in the construction of overall $l_p$-type measures \cite{baek20}. Hellinger distance is a known $f$-divergence \cite{pollard02}, but (as far as we know) has not been used in the present context before. It is an actual metric in the space of probability distributions; in particular, if $d^2_{nm}(\M)=0$ for some $n,m$, then $p^\M_n(i) = p_m^\M(i)$ for all $i\in \Omega_\M$. Furthermore, $0\leq d^2_{nm}(\M)\leq 1$ by a simple application of the classical Schwarz inequality.

We say that $\M$ is \emph{maximally coherent} if $\coh_{nm}(\M)=1$ for all $n,m=1,\ldots,d$.
The following proposition summarises the basic properties of the entry-wise coherence:
\begin{appendixproposition}\label{maxcohprop} Let $\M$ be an observable.
\begin{itemize}
\item[(a)] (Bounds). $0\leq \coh_{nm}(\M)\leq 1-d^2_{nm}(\M)\leq 1$ for all $n,m$.
\item[(b)] (Zero coherence). $\coh_{nm}(\M) =0$ for all $n\neq m$, if and only if $\M$ is incoherent.
\item[(c)] (Maximal coherence). The following are equivalent:
\begin{enumerate}
\item[(i)] $\M$ is maximally coherent;
\item[(ii)] $\M$ has has rank one, and $p_n^\M$ is the same probability distribution for each $n=1,\ldots,d$;
\item[(iii)] There is a probability distribution $i\mapsto p(i)$ on $\Omega_\M$, and a sequence of unit vectors $\psi_i\in \mathcal H$ with $|\langle n|\psi_i\rangle|^2=d^{-1}$ for each $n=1,\ldots,d$, $i\in \Omega_\M$, such that $\M(i) =p(i) d |\psi_i\rangle\langle \psi_i|$ for each $i\in \Omega_\M$.
\end{enumerate}
If $\M$ has exactly $d$ outcomes, then $\M$ is maximally coherent if and only if $\M$ is a MUB to the incoherent basis.
\end{itemize}
\end{appendixproposition}
\begin{proof}
Denote $\psi_i^{n} = \sqrt{\M(i)} |n\rangle$, for each $n=1,\ldots, d$, $i\in \Omega_\M$. Then $\|\psi_i^n\|^2=p_n^\M(i)$, and hence $\sum_{i\in \Omega_\M} \|\psi_i^n\|^2=\sum_i p_n^\M(i) = 1$ for each $n$ by the normalisation of the observable $\M$. To prove (a) we use the Cauchy-Schwarz inequality:
\begin{align*}
\coh_{nm}(\M) &= \sum_i |\langle n|\M(i) |m\rangle| = \sum_i |\langle \psi_i^{n}|\psi_i^m\rangle|\\
&\leq  \sum_i \|\psi_i^{n}\|\||\psi_i^m\|=1-d_{nm}^2(\M)\leq 1. 
\end{align*}
If $\coh_{nm}(\M) =0$ for all $n\neq m$ we have $\langle n|\M(i)|m\rangle = 0$ for all $n\neq m$, and hence $\M$ is incoherent; this proves (b).

To prove (c), assume (i), so that $\coh_{nm}(\M)=1$ for each pair $(n,m)$. Then the second inequality in the above calculation is saturated, so $d_{nm}^2(\M)=0$ for all $n\neq m$, which implies $p(i):=p^\M_n(i) = p_m^\M(i)$ for each pair $n,m=1,\ldots,d$, and all $i\in \Omega_\M$, so $\|\psi_i^n\|^2=p(i)$ for each $n,i$ (i.e. the norm only depends on $i$). Note that $p(i)>0$ for each $i$ (since otherwise $\M(i)=0$). Also the first inequality is saturated, that is, $\sum_i ( \|\psi_i^{n}\|\||\psi_i^m\|-|\langle \psi_i^{n}|\psi_i^m\rangle|)=0$, so $ |\langle \psi_i^{n}|\psi_i^m\rangle|=\|\psi_i^{n}\|\||\psi_i^m\|$ for all $n,m,i$, as each term in the sum is nonnegative. Hence the Cauchy-Schwarz inequality is saturated for each pair $\psi^n_i,\psi_i^m$, so $\psi_i^n =c_i^{nm}\psi_i^m$ for some constants $c_i^{nm}\in \mathbb C$ which must have modulus one as $\|\psi_i^n\|^2 = \|\psi_i^m\|^2=p(i)$. Define $\psi_i := p(i)^{-\frac 12} \psi_i^{1}$ for each $i\in \Omega_\M$. Then $\|\psi_i\|=1$, and $\sqrt{\M(i)} |n\rangle=\psi_i^n = c_{i}^{n1}\psi_i^{1}=c_{i}^{n1}\sqrt{p(i)}\psi_i$ for all $n,i$, showing that $\sqrt{\M(i)}$ (and hence also $\M(i)$) has rank one with range spanned by $\psi_i$. Hence (ii) holds. Assume now (ii). Since $\M$ has rank one we can write $\M(i) = p(i) d|\psi_i\rangle\langle\psi_i|$ where $p(i) := {\rm tr}[\M(i)]/d$ is a probability distribution and $\psi_i$ is a unit vector for each $i\in\Omega_\M$. Since $p^\M_n=p^\M_m$ for all $n,m$, we must have $p(i)=\frac 1d\sum_m p^\M_m(i)=p^\M_n(i)= \langle n|\M(i)|n\rangle =p(i) d |\langle \psi_i|n\rangle|^2$ for all $n,i$, which shows that $|\langle \psi_i|n\rangle|^2=d^{-1}$ for all $n,i$. Hence (iii) holds. Finally, assuming (iii) we easily check that $\M$ has maximal coherence, i.e. (i) holds, and we have established the equivalences in (c). The last claim follows immediately from (iii) and the well-known fact (which is easy to prove) that any rank one observable in dimension $d$ with $d$ outcomes is necessarily a basis observable.
\end{proof}
As noted above, any observable mutually unbiased to the incoherent basis is incoherent. Moreover, any refinement of a such an observable is maximally coherent: if $q_k(i)_{i\in\Omega_k}$ defines a probability distribution for each $k=1,\ldots,d$ (where $\Omega_k$ are distinct sets), and $\{\psi_k\mid k=1,\ldots,d\}$ is mutually unbiased to the incoherent basis, let $\Omega = \cup_{k}\Omega_k$, and $\M(i) := q_k(i)|\psi_k \rangle\langle\psi_k|$ whenever $i\in \Omega_k$. Another class of maximally coherent observables is given as follows: take any pair of MUBs $\{\phi_k\mid k=1,\ldots N\}$ and $\{\xi_k\mid k=1,\ldots,N\}$ in a larger Hilbert space $\mathcal M=\mathbb C^N$, and define an isometry $V:\mathcal H\to \mathcal M$ by $V|n\rangle = \xi_n$ for $n=1,\ldots,d$. Then $\M(i) := |V^*\phi_i\rangle\langle V^*\phi_i|$, $i\in \Omega_\M:=\{1,\ldots, N\}$ is a maximally coherent observable with the distribution $p(i) = 1/N$ in the above proposition. An explicit example of this type is obtained by taking $\Omega_\M=\{0,\ldots,N\}$, $p(i) =1/N$ and
$\ket{\psi_i}=\frac{1}{\sqrt{d}}\sum_{j=0}^{d-1}\omega_N^{ji}\ket{j}$
for $i=0,\ldots,N-1$, where $\omega_N=e^{2\pi i/N}$.

\section{Dilation theory}\label{A:dilations}

We review here briefly some well-known aspects of dilation theory of quantum channels and observables (see, e.g.,\ \cite{QM}), applied to our framework introduced in the main text.

First recall that the \emph{Naimark dilation} of an observable $\M=(\M(i))_{i\in \Omega_\M}$ on a Hilbert space $\mathcal H$ is a projection valued observable $\A$ on a larger Hilbert space $\hi_\oplus$ 
such that $\M(i) = J^\dagger \A(i)J$ for all $i$, where $J:\mathcal H\to \mathcal H_\oplus$ is an isometry, i.e.\ $J^\dagger J=\id$. The dilation is \emph{minimal}, if $\mathcal H_\oplus={\rm span}\{\A(i)J\varphi\mid \varphi\in \mathcal H, \, i\in\Omega_\M\}$. 
The following is a basic joint measurability result:

\begin{appendixtheorem}[\cite{Optimal}]\label{JMthm} Let $\F=(\F(j))_{j\in\Omega_\F}$ be any observable jointly measurable with $\M$ and $(\hi_\oplus,\A,J)$ a minimal Naimark dilation of $\M$. Then each joint observable $\G$ of $\M$ and $\F$ is of the form $\G(i,j) = J^\dagger \A(i) \B(j)J$ where $\B$ is a unique  POVM of $\hi_\oplus$ such that $[\A(i),\B(j)]=0$ for all $i\in \Omega_\M$, $j\in \Omega_\F$.
\end{appendixtheorem}

Any quantum channel $\Lambda$ of $\hi\simeq\mathbb C^d$ has a {\it minimal Stinespring dilation,} i.e.\ its Heisenberg picture (a completely positive unital map on the matrix algebra $M_d(\mathbb C)$)
 can be written in the form $\Lambda^*(X)=J^\dagger (X\otimes \id)J$, $X\in M_d(\mathbb C)$,
 where $J:\,\hi\to\hi\otimes\ki$ is an isometry, $\ki$ a Hilbert space (an ancilla) and the vectors 
 $(X\otimes\id)J\psi$, $X\in M_d(\mathbb C)$, $\psi\in\hi$, span $\hi\otimes\ki$ \cite{paulsen02}.
 It follows from the Radon-Nikodym theorem of completely positive maps \cite{raginsky03} that any instrument $(\I_i)_{i\in \Omega}$ whose channel is $\Lambda$ has the form $\I^*_i(X) = J^\dagger (X\otimes \F(i))J$ where $\F=(\F(i))_{i\in \Omega}$ is a (unique) POVM.

We now specialise to our case, with a channel $\Lambda_C$ given by a PSD matrix $C$ with unit diagonal, through Hadamard multiplication $\Lambda_C^*(X)=C*X$. Since $C\geq 0$, we may write $c_{nm}=\langle \eta_n|\eta_m\rangle$ where $\eta_n$ are unit vectors in a Hilbert space $\mathcal K$ with dimension equal to the rank of $C$, that is, $\mathcal K={\rm span} \{\eta_n\mid n=1,\ldots,d\}$ (see, e.g.,\ \cite{helm09}). These vectors constitute the minimal Stinespring dilation $\Lambda^*_C(X) = J^\dagger (X\otimes \id)J$ of $\Lambda_C$, where the isometry is defined by $J|n\rangle = |n\rangle\otimes \eta_n$ (note that $|m\rangle\otimes\eta_n=(X\otimes\id)J|n\rangle$ where $X=\kb m n$).
Then any GII $(\I_j)_{j\in \Omega}$ with channel $\Lambda_C$ has the form 
$$
\I^*_j(X) = J^\dagger (X\otimes \F(j))J=\sum_{n,m}\langle n|X|m\rangle\langle \eta_n|\F(j)\eta_m\rangle\kb n m
$$ where $\F=(\F(j))_{j\in \Omega}$ is a POVM of $\ki$.
This gives us the Hadamard form $\I^*_j(X) = C(j)* X$ used in the main text, with the matrix $C(j)$ given by
\begin{equation}\label{GIIdilation}
c_{nm}(j)=\langle \eta_n|\F(j)\eta_m\rangle.
\end{equation}
We stress that \emph{every} instrument with channel $\Lambda_C$ has this form. In particular, if $\{C(j)\}$ is any collection of PSD matrices with $\sum_j C(j)=C$, then $C(j)$ can be written as \eqref{GIIdilation} for some (unique) $\F$.

\section{Proofs of Theorems \ref{mubthm} and \ref{mainthm}}\label{A:mainthm}

Recall that in Theorem \ref{mubthm} we let $\Q_0$ be any fixed maximally coherent observable, with outcome set $\Omega_0$ (assumed to be minimal). By Prop. \ref{maxcohprop}, we can write
\begin{equation}\label{maxcohQ}
\Q_0(i) = p(i) d |\psi_i\rangle\langle\psi_i|
\end{equation}
where $i\mapsto p(i)$ is a probability distribution on $\Omega_0$, and $\psi_i$ are unit vectors such that $|\langle \psi_i|n\rangle|^2=d^{-1}$ for all $n,i$. Let $\mathcal M$ be a copy of $\mathbb C^{|\Omega_0|}$, and note that $\langle n|\Q_0(i)|n\rangle = p(i)$ for all $n=1,\ldots,d$.

The proofs are based on the dilation theory described above; let $\eta_n$ be vectors such that $c_{nm}=\langle \eta_n|\eta_m\rangle$ as in Appendix \ref{A:dilations}.

{\it Proof of Theorem \ref{mubthm}}. Fix a basis $\{|i\rangle\mid i\in \Omega_{0}\}$ of $\mathcal M$, and define an isometry $V:\mathcal H\to \mathcal M \otimes \mathcal K$ via
$$
V|n\rangle = \sum_{i\in \Omega_{\Q_0}} \sqrt{p(i)d}\,\langle \psi_i|n\rangle\, |i\rangle \otimes \eta_n.
$$
Then $C* \Q_0(i) = V^*(|i\rangle\langle i|\otimes \id_{\mathcal K})V$. Hence, this is a Naimark dilation of the observable $C*\Q_0$. Since ${\rm span}\{(|i\rangle\langle i|\otimes \id)V|n\rangle \mid n\in\{ 1,\ldots, d\},i\in\Omega_{0} \}= {\rm span}\{|i\rangle\otimes \eta_n\mid n\in\{1,\ldots, d\}, i\in \Omega_0\}=\mathcal M\otimes\mathcal K$ $(=\hi_\oplus)$, the dilation is minimal.

Now assume that $C*\Q_0$ is jointly measurable with an incoherent observable $\P$. Since the above dilation is minimal, Thm.\ \ref{JMthm} applies: $\P$ must be of the form $\P(j) = V^\dagger \F(j)V$, $j\in\Omega_\P$, where $[\F(j), |i\rangle\langle i|\otimes \id_\mathcal K]=0$ for all $i,j$. This implies that for each $i$ there is a POVM $\F_i=(\F_i(j))_{j\in \Omega_\P}$ on $\mathcal K$, such that $\F(j) = \sum_i |i\rangle\langle i|\otimes \F_i(j)$. Furthermore, according to Thm.\ \ref{JMthm}, $\P$ and $C*\Q_0$ have a joint observable
\begin{align*}
\tilde \G(i,j) &= V^\dagger \F(j)(|i\rangle\langle i|\otimes \id)V= V^\dagger (|i\rangle\langle i|\otimes \F_i(j))V\\
&=\sum_{n,m} \langle n|\Q_0(i)|m\rangle \langle \eta_n |\F_i(j)\eta_m\rangle |n\rangle\langle m|
\end{align*}
and hence $\P$ must have the form
\begin{align}\label{crucial}
\P(j) &= \sum_i \tilde \G(i,j)\nonumber= \sum_{n,m}\sum_i \langle \eta_n |\F_i(j)\eta_m\rangle \langle n|\Q_0(i)|m\rangle |n\rangle\langle m|\nonumber\\
&=\sum_{n}\sum_i \langle \eta_n |\F_i(j)\eta_n\rangle \langle n|\Q_0(i)|n\rangle |n\rangle\langle n|\nonumber\\
& = \sum_{n}\big\langle \eta_n \big|\sum_i p(i)\F_i(j)\eta_n\big\rangle |n\rangle\langle n|.
\end{align}
In the third step we have used the assumption that $\P$ is incoherent (so there are no off-diagonal elements), and in the fourth step the maximal coherence condition $\langle n|\Q_0(i)|n\rangle =p(i)$. We now define, for each $j$, a matrix $C(j)$ by $c_{nm}(j) := \langle \eta_n |\A(j)|\eta_m\rangle$, where $\A := \sum_i p(i)\F_i$ is a POVM by convexity. Therefore we have $C(j)\geq 0$ and $\sum_j c_{nm}(j)= \langle \eta_n|\eta_m\rangle =c_{nm}$, that is, $C(j)$ form a GII whose channel is $\Lambda_C$. Finally, by the computation \eqref{crucial}, $\P(j) = \sum_n  \langle \eta_n |\A(j)\eta_n\rangle |n\rangle\langle n| = \sum_n c_{nn}(j) |n\rangle\langle n|$, showing that the observable of this GII is precisely $\P$. Applying the GII to any observable $\Q$ we get a joint observable $\G(i,j) = C(j)*\Q(i)$ for $C*\Q$ and $\P$, as $\sum_i \G(i,j) = C(j) *\id = \P(j)$ and $\sum_j \G(i,j) = C*\Q(i)$. Hence $C*\Q$ and $\P$ are jointly measurable. This completes the proof of Thm. \ref{mubthm}.

The crucial point of the proof is the computation \eqref{crucial}; one can readily see how the two strong assumptions, $\P$ incoherent and $\Q_0$ maximally coherent, fit together rather neatly to form the \emph{single} dilation POVM $\A$.

{\it Proof of Theorem \ref{mainthm}}. If (i) holds then $C*\Q$ is jointly measurable with $\P$ for all $\Q$, so in particular for $\Q_0$. By the above proof we obtain matrices $C(j)$ satisfying \eqref{GII0}, so (ii) holds. Conversely, if such matrices exist (that is, (ii) holds), the observable $\G(i,j) = C(j)*\Q(i)$ defined in the above proof is a joint observable for $\P$ and $C*\Q$ for any observable $\Q$, hence $\P\in \mathcal C_C$, i.e., (i) holds.

\section{Reduction by symmetry}\label{A:symmetry}

Here we develop in detail the theory of covariance systems for a $d\times d$ coherence matrix $C$. Recall that the aim is to characterise the set $\mathcal C_C$ of incoherent observables $\P$ for which there is a GII with GIO $\Lambda_C$ and observable $\P$. The idea is that symmetries in the coherence pattern can be used to simplify the problem, and single out relevant incoherent observables.

As above, we make use of the dilation $c_{nm}=\langle \eta_n|\eta_m\rangle$ (see Appendix \ref{A:dilations}). For each pair $(n,m)$ we write $n\sim_C m$ when $|c_{nm}|=1$. This implies that $|\langle \eta_n|\eta_m\rangle| = 1=\|\eta_n\|\|\eta_m\|$, i.e.\ the Cauchy-Schwarz inequality is saturated for this pair of unit vectors, and hence $\eta_n = e^{i\theta}\eta_m$ for some $\theta\in \mathbb R$. Now if $|c_{nm}|=1$ and $|c_{mk}|=1$ then $\eta_n = e^{i\theta}\eta_m$ and $\eta_m = e^{i\theta'}\eta_k$, so $\eta_n = e^{i(\theta+\theta')} \eta_k$, which implies $|c_{nk}|=1$. Hence the relation $n\sim_C m$ is transitive, and since $|c_{nm}|=|c_{mn}|$ for each pair $(n,m)$, it is also symmetric, so an equivalence relation on the set $\{1,\ldots, d\}$. Hence the set splits into a union $\{1,\ldots ,d\}=\cup_{k=1}^{n_C} I^C_k$ of $n_C$ equivalence classes $I^C_k$ (unique up to ordering). We let $\Omega_C=\{1,\ldots, n_C\}$.

Now we pick from each equivalence class $I^C_k$ one fixed representative $n_k\in I^C_k$, and let $\tilde \eta_k:=\eta_{n_k}$ for each $k$; then for each $n=1,\ldots, d$ there is a unique phase factor $e^{i\theta_n}$ so that $\eta_n = e^{i\theta_n} \tilde \eta_{k}$ where $I^C_k$ is the class of $n$. Now define an $n_C\times n_C$ matrix $\tilde C$ by $\tilde c_{kk'} := \langle \tilde\eta_{k}|\tilde\eta_{k'}\rangle$. By construction, $\tilde C$ is a structure matrix, and we observe that $c_{nm} = \langle \eta_n|\eta_m\rangle = e^{-i(\theta_n-\theta_m)} \tilde c_{kk'}$ whenever $n\in I^C_k$ and $m\in I^C_{k'}$. Let $D$ be the matrix $d_{nm} = e^{-i(\theta_n-\theta_m)}$; this is symmetric rank-1, hence the structure matrix of a unitary GIO $\Lambda^*_D(X) = U_DXU^{\dagger}_D$ where $U_D$ is the diagonal unitary with phases $e^{-i\theta_n}$ on the diagonal. We then let $\tilde{\mathcal H}$ be the Hilbert space with incoherent basis $\{|k\rangle\}_{k\in \Omega_C}$ and define $L:\tilde{\mathcal H}\to \mathcal H$ by $$L|k\rangle = \sum_{n\in I^C_k} |n\rangle.$$

We then obtain the decomposition
$C = D* (L\tilde CL^\dagger )$, so that $C \simeq L\tilde CL^\dagger $, and see that the $(n,m)$ entry of the structure matrix $L\tilde CL^\dagger $ is equal to $\tilde c_{kk'}$ for all $n\in I^C_k$, $m\in I^C_{k'}$, that is, only depends on the classes of $n$ and $m$. In other words, after the unitary GIO is factored out, the remaining channel compresses into the GIO $\Lambda_{\tilde C}$ on a $n_C$-dimensional system. By construction, this channel is unique up to diagonal unitaries, corresponding to different choices of the representatives $n_k$.

Assuming $\P$ is adapted, we can compress it into an incoherent observable $\tilde \P$ on the $n_C$-dimensional system, by setting $p^{\tilde \P}_k := p^{\P}_n$ for any $n\in I^C_k$, so that $\P(j) = \mathcal L(\tilde \P(j))$ where $\mathcal L: {\rm span} \{|k\rangle\langle k|\mid k\in \Omega_C\} \to {\rm span} \{|n\rangle\langle n|\mid n=1,\ldots,d\} $ links the two diagonal algebras ``incoherently'':
$$\mathcal L(|k\rangle\langle k|) = \sum_{n\in I^C_k} |n\rangle\langle n|.$$
We can now prove our first reduction result:

\

\begin{appendixproposition}\label{adapted}
There is a GIO $\Lambda_{\tilde C}$ acting on $\tilde{\mathcal H}$ such that $C\simeq L\tilde CL^\dagger$. Then $\mathcal C_{C} = \{ \mathcal L(\P(\cdot) )\mid \P\in \mathcal C_{\tilde C}\}$.
\end{appendixproposition} 
\begin{proof} The decomposition $C\simeq L\tilde CL^\dagger$ was constructed above. To prove the second claim, assume first that $\P'\in \mathcal C_C$, and let $C(j)$ form a GII with observable $\P'$ and GIO $\Lambda_C$, so that $c_{nm}(j) = \langle \eta_n| \F(j) |\eta_m\rangle$ for some POVM $\F$ on the dilation space $\ki$, with $\langle \eta_n| \F(j) |\eta_n\rangle=p^{\P'}_n(j)$ for each $n$. Now the $n_C\times n_C$-matrices $\tilde C(j)$ defined by $\tilde c_{kk'}(j) := \langle \tilde \eta_k| \F(j) |\tilde \eta_{k'}\rangle$ form a GII with channel $\Lambda_{\tilde C}$. In order to find the corresponding observable we compute $\tilde c_{kk}(j) = \langle \tilde \eta_k| \F(j) |\tilde \eta_{k}\rangle = \langle \eta_n| \F(j) |\eta_{n}\rangle = p^{\P'}_n(j)$ for \emph{any} $n\in I^C_k$. Hence $p^{\P'}_n$ does not depend on the choice of $n\in I^C_k$, so $\P'$ is adapted to $C$, and $\P' =\mathcal L(\P)$ where $\P$ is defined by $p_k^{\P}(j):=\tilde c_{kk}(j)$. This shows that the observable of this GII is $\P$, so $\P\in \mathcal C_{\tilde C}$. Conversely, if $\P'=\mathcal L(\P)$ with $\P\in \mathcal C_{\tilde C}$ then there is a GII $\tilde C(j)$ with $\sum_j \tilde c_{kk'}(j) = \tilde c_{kk'}$ and $\tilde c_{kk}(j) = p^{\P}_k(j)$. We then define $c_{nm}(j):=e^{-i(\theta_n-\theta_m)}\tilde c_{kk'}(j)$ whenever $(n,m)\in I^C_k\times I^C_{k'}$. This is a GII for which $\sum_j c_{nm}(j) = e^{-i(\theta_n-\theta_m)}\tilde c_{kk'}=c_{nm}$, and $c_{nn}(j)=\tilde c_{kk}(j)=p^{\P}_k(j) = p^{\P'}_n(j)$ regardless of the choice of $n\in I^C_k$. Hence $\P'\in \mathcal C_{C}$. This completes the proof.
\end{proof}

Now let $S_d$, the group of permutations of $\{1,\ldots,d\}$, act on $\mathcal H$ via $U_\pi|n\rangle=|\pi(n)\rangle$, and recall from the main text, the symmetry group
$$
G_C = \{\pi\in S_d\mid U_\pi^\dagger  C U_\pi \simeq C\}.
$$
By the definition of $\simeq$, $G_C$ consists of exactly those permutations $\pi\in S_d$ for which there exists a unitary GIO with matrix $D$ such that $U_\pi^\dagger  C U_\pi = D *C$. So $\pi\in G_C$ iff there exist phase factors $u_n(\pi)$, $n=1,\ldots,d$, such that
\begin{equation}\label{phases}
c_{\pi(n),\pi(m)} = u_{n}(\pi) c_{nm} \overline{u_m(\pi)}
\end{equation}
for each $n,m$. In what follows we assume for simplicity that $c_{nm}\neq 0$ for all $n,m$. Then for each $\pi$ the coefficients $u_n(\pi)$ are uniquely determined up to an overall ($\pi$-dependent) phase factor, which we choose by setting $u_{n_0}(\pi)=1$ for a fixed $n_0$. We can then construct $u_n(\pi)$ explicitly from the entries of $C$:
\begin{equation}\label{defu}
u_n(\pi) = c_{\pi(n),\pi(n_0)}/c_{nn_0}, \quad n=1,\ldots,d.
\end{equation}
We then define, for each $\pi\in G_C$ and $n=1,\ldots,d$, a unitary operator $W_\pi$ on $\hi$ by $$W_\pi|n\rangle = u_n(\pi) |\pi(n)\rangle, \quad n=1,\ldots,d,$$
so that (by \eqref{phases}) we may write
$$
G_C = \{\pi\in S_d\mid W_\pi^{\dagger} C W_\pi =C\}.
$$
The following result shows that $W_\pi$ appropriately reflects the symmetries of $C$ on the Hilbert space level:
\begin{appendixproposition} $G_C$ is a permutation group (i.e.\ a subgroup of $S_d$), and $\pi\mapsto W_\pi$ is a projective unitary representation of $G_C$ with multiplier $(\pi,\pi')\mapsto u_{\pi'(n_0)}(\pi)$. If each entry of $C$ is real positive, then $W_\pi = U_\pi$.
\end{appendixproposition}
\begin{proof} Let $\pi,\pi'\in G_C$. Then \eqref{phases} holds for both, so
\begin{align*}
c_{\pi\pi'(n),\pi\pi'(m)} &= u_{\pi'(n)}(\pi) c_{\pi'(n),\pi'(m)} \overline{u_{\pi'(m)}(\pi)}\\
&= u_{\pi'(n)}(\pi) u_{n}(\pi') c_{nm} \overline{u_{m}(\pi')} \overline{u_{\pi'(m)}(\pi)}
\end{align*}
for all $n,m$, showing that \eqref{phases} holds also for $\pi\pi'$. Hence $\pi\pi'\in G_C$, and since $G_C$ is finite, this implies that $G_C$ is a subgroup. Taking $m=n_0$ and using \eqref{defu} we find $$u_{n}(\pi\pi') =  u_{\pi'(n)}(\pi) u_{n}(\pi')\overline{u_{\pi'(n_0)}(\pi)},$$
which reads $W_\pi W_{\pi'} = u_{\pi'(n_0)}(\pi)\, W_{\pi\pi'}$. Since clearly $W_e=\id$, the second claim follows. Finally, if each entry of $C$ is real positive, then $u_n(\pi)=1$ for all $\pi$ and $n$, and we have simply $W_\pi = U_\pi$.
\end{proof}

It is clear that $\{I^C_k\}$ forms a block system for the group $G_C$: for each $\pi\in G_C$ we have $|c_{nm}|=1$ iff $|c_{\pi(n),\pi(m)}|=1$, so $\pi$ moves each class as a whole, $\pi(I^C_k) = I^C_{\phi(\pi)(k)}$ for a unique $\phi(\pi)\in G_{\tilde C}$, where $\tilde C$ is the reduced $n_C\times n_C$ GIO matrix. The map $\phi:G_C\to G_{\tilde C}$ is a homomorphism with $\phi(G_C)\leq G_{\tilde C}$ consisting of permutations between classes of the same size. This structure is unique up to an irrelevant overall permutation of $\{1,\ldots, n_C\}$, fixed by the labelling of $I^C_k$.

Next, recall that given any subgroup $G\leq G_C$, a $G$-covariant incoherent observable $\P$ is one with outcome set $\Omega_C$ satisfying $p^{\P}_{\pi^{-1}(n)}(j) = p^{\P}_n(\phi(\pi)(j))$ for each $\pi\in G$, $j\in \Omega_C$ and $n\in \{1,\ldots d\}$. It is convenient to write this condition equivalently using the representation $W_\pi$ as
\begin{equation}\label{ADcov}
W_\pi \P(j) W_\pi^{\dagger} = \P(\phi(\pi)(j)), \, \pi\in G,\, j \in \Omega_C.
\end{equation}
Note that the reduced matrix $\tilde C$ obviously does not reduce further, i.e.\ $n_{\tilde C}=n_C$ (each equivalence class is a singleton). Hence for any subgroup $\tilde G\leq G_{\tilde C}\leq S_{n_C}$, the $\tilde G$-covariant observables $\P$ are given by \eqref{ADcov} with $d$ replaced by $n_C$ and $\phi={\rm Id}$.

We recall that $\mathcal C_C[G]$ is the set of all $G$-covariant incoherent observables in $\mathcal C_C$, and $\mathcal C_C^{\rm sym}=\mathcal C_C[G_C]$. The symmetry constraint \eqref{ADcov} can be naturally formulated in the GII level: we call a GII $C(j)$ \emph{$G$-covariant} if
\begin{equation}\label{ADcovC}
W_\pi C(j) W_\pi^{\dagger} = C(\phi(\pi)(j)), \, \pi\in G,\, j \in \Omega_C.
\end{equation}
Notice that here the matrices $C(j)$ are not diagonal, so we need to state the condition using the representation $W_\pi$. The following result shows that the SDP \eqref{GII0} in the main text can be supplemented by an extra symmetry constraint if $\P$ is $G$-covariant:

\begin{appendixproposition}\label{GIIsymmetry}
Any $\P\in \mathcal C_C[G]$ has a $G$-covariant GII.
\end{appendixproposition}
\begin{proof}
To prove the claim, let $\P\in \mathcal C_C[G]$. Then there is a GII $C(j)$ with $\sum_j C(j) = C$ and $c_{nn}(j) =p^{\P}_n(j)$ for all $n$. We define
$$
C'(j) := \frac{1}{|G|} \sum_{\pi\in G} W_\pi^{\dagger} C(\phi(\pi)(j)) W_\pi.
$$
This is essentially the ``averaging argument'' often used in the context of symmetry constraints for joint measurability \cite{werner04,carmeli05,carmeli12}, except that now we apply it at the level of structure matrices as opposed to POVM elements. Now $C'(j)\geq 0$, so it defines a GII, which is $G$-covariant, as
\begin{align*}
W_{\pi_0}C'(j) W^{\dagger}_{\pi_0} &= \frac{1}{|G|} \sum_{\pi\in G} W_{\pi\pi^{-1}_0}^{\dagger} C(\phi(\pi)(j)) W_{\pi\pi^{-1}_0} \\
&= \frac{1}{|G|} \sum_{\pi\in G} W_{\pi}^{\dagger} C(\phi(\pi\pi_0)(j)) W_{\pi}\\
&= \frac{1}{|G|} \sum_{\pi\in G} W_{\pi}^{\dagger} C(\phi(\pi)\phi(\pi_0)(j)) W_{\pi}\\
&=C'(\phi(\pi_0)(j))
\end{align*}
for each $\pi_0\in G$. Here we used the fact that $\pi\mapsto W_\pi$ is a (projective) representation, and $\phi$ is a homomorphism. Furthermore,
\begin{align*}
\sum_j C'(j) &=  \frac{1}{|G|} \sum_{\pi\in G} W_\pi^{\dagger} \sum_j C(\phi(\pi)(j)) W_\pi\\
&= \frac{1}{|G|} \sum_{\pi\in G} W_\pi^{\dagger} C W_\pi = \frac{1}{|G|} \sum_{\pi\in G} C = C,
\end{align*}
as $G$ is a subgroup of $G_C=\{\pi\in S_d\mid W_\pi C W_\pi^{\dagger}=C\}$. Finally,
\begin{align*}
c'_{nn}(j) &= \frac{1}{|G|} \sum_{\pi\in G} \langle n| W_\pi^{\dagger} C(\phi(\pi)(j)) W_\pi|n\rangle\\
&= \frac{1}{|G|} \sum_{\pi\in G} c_{\pi(n),\pi(n)}(\phi(\pi)(j))\\
&= \frac{1}{|G|} \sum_{\pi\in G} p^{\P}_{\pi(n)}(\phi(\pi)(j))\\
&= \frac{1}{|G|} \sum_{\pi\in G} p^{\P}_{n}(j)=p^{\P}_n(j),
\end{align*}
because $\P$ is $G$-covariant. Therefore, $C'(j)$ satisfies eq. \eqref{GII0} in the main text, and is $G$-covariant.
\end{proof}

We also remark that the GII matrices can always be chosen real if $C$ is a real matrix (independently of permutation symmetry). In fact, if \eqref{GII0} holds for matrices $C(j)$, we can define $C^{\rm re}(j) = \frac12(C(j) + C(j)^T)$; then $C^{\rm re}(j)\geq 0$ since transpose preserves positivity, and $C^{\rm re}(j)$ is a real symmetric matrix since $c^{\rm re}_{nm}(j) = {\rm Re}\, c_{nm}(j)$ (as $C(j)=C(j)^*$). Since $C$ is real we therefore still have $\sum_j C^{\rm re}(j) = C$, and since the diagonal of $C(j)$ is real in any case, it coincides with the diagonal of $C^{\rm re}(j)$. Hence the matrices $C^{\rm re}(j)$ fulfil \eqref{GII0} as well.

We now prove the main reduction result, which in the main text was stated in Eq. \eqref{symreduction}:
\begin{appendixproposition} \label{symmetryprop1}
$\mathcal C_C^{\rm sym}=\{\mathcal L(\P(\cdot))\mid \P\in \mathcal C_{\tilde{C}}[\phi(G_C)]\}$.
\end{appendixproposition}
\begin{proof} Let $\P\in \mathcal C_{\tilde C}[\phi(G_C)]$. Hence $\P$ is $\phi(G_C)$-covariant and $\P\in \mathcal C_{\tilde C}$. Now define $\P':= \mathcal L(\P)$. Then $\P'\in \mathcal C_{C}$ by Prop.\ \ref{adapted}. Note that $\P'$ still has $n_C$ outcomes, but lives in dimension $d$; explicitly, $\P'(j) =\sum_{k=1}^{n_C} p^{\P}_k(j) \sum_{n\in I_k^C} |n\rangle\langle n|$. The following rearrangement now shows that $\P'$ is $G_C$-covariant:
\begin{align*}
W_\pi \P'(j) W_\pi^{\dagger} &= \sum_{k=1}^{n_C} p^{\P}_k(j) \sum_{n\in I_k^C} |\pi(n)\rangle\langle \pi(n)|\\
&= \sum_{k=1}^{n_C} p^{\P}_k(j) \sum_{\pi^{-1}(n)\in I_k^C} |n\rangle\langle n|\\
&= \sum_{k=1}^{n_C} p^{\P}_k(j) \sum_{n\in I_{\phi(\pi)(k)}^C} |n\rangle\langle n|\\
&= \sum_{k=1}^{n_C} p^{\P}_{\phi(\pi)^{-1}(k)}(j) \sum_{n\in I_{k}^C} |n\rangle\langle n|\\
&= \sum_{k=1}^{n_C} p^{\P}_{k}(\phi(\pi)(j)) \sum_{n\in I_{k}^C} |n\rangle\langle n|
= \P'(\phi(\pi)(j)).
\end{align*}
Hence $\P'\in \mathcal C_C^{\rm sym}$. Conversely, if we pick a $\P'\in\mathcal C_C^{\rm sym}$ then by Prop.\ \ref{adapted} we can write it as $\P':= \mathcal L(\P)$ for some $\P\in \mathcal C_{\tilde C}$, and reverse the rearrangement to show that $p^{\P}_{\phi(\pi)^{-1}(k)}(j)= p^{\P}_{k}(\phi(\pi)(j))$, i.e.\  $\P$ is $\phi(G_C)$-covariant and hence $\P\in \mathcal C_{\tilde C}[\phi(G_C)]$. This completes the proof.\end{proof}

Finally, we modify the robustness idea described above to account for symmetry: instead of Eq. \eqref{whitenoise} we use the ``canonical'' $G_C$-covariant observable $\P^C(j):=\mathcal L(|j\rangle\langle j|)$, $j\in \Omega_C$, and the line $\P^C_{\alpha} = \alpha \P^C +(1-\alpha)\P_{\rm dep}^C$, where $\P_{{\rm dep}}^C(j)=d^{-1}{\rm tr}[\P^C(j)]\id=d^{-1} |I^C_j|\id$. Noting that $\P^C\notin \mathcal C_C^{\rm sym}$, we use Prop.\ \ref{symmetryprop1} to set 
\begin{align}\label{CI}
\alpha_C &:= \max \{\alpha>0 \mid \P^C_{\alpha}\in \mathcal C_C\}\\
 &=\max \{\alpha>0 \mid \P^C_{\alpha}\in \mathcal C^{\rm sym}_C\}\nonumber\\
&=\max\{\alpha>0 \mid  \alpha |j\rangle\langle j| + (1-\alpha) d^{-1} |I^C_j| \id \in \mathcal C_{\tilde C}^{\rm sym}\}, \nonumber
\end{align}
so $\alpha_C$ is the proportion of the line where coherence does \emph{not} sustain incompatibility. Note that in the nondegenerate case ($|c_{nm}|<1$ for all $n\neq m$) we have $\P^C_\alpha=\P_\alpha$, so Eq. \eqref{CI} is (by Thm. \ref{mubthm}) consistent with $\alpha_\M$ defined after Eq. \eqref{whitenoise} in the main text, when $\M =C*\Q_0$ where $\Q_0$ is maximally coherent.

\section{Uniform coherence with negative entries}\label{A:uniform}

Let $C$ be a coherence matrix with full symmetry, i.e.\ $G_C=S_d$. Then $c_{nm}=\lambda$ for all $n\neq m$, for some $\lambda\in \mathbb R$, i.e.\ all coherences are equal. From its eigenvalues one sees that $C$ defines a GIO iff $-(d-1)^{-1}\leq \lambda \leq 1$. Any $G_C$-covariant $\P$ has $p^{\P}_{n}(j) = q$ for all $n,j$, $j\neq n$, for a fixed $q$; writing $q=(1-\alpha)/d$ we see that $\P= \P^C_\alpha=\P_\alpha$ for some $-(d-1)^{-1}\leq \alpha \leq 1$. Hence these families are naturally motivated by symmetry considerations. Recall from the main text (Prop. \ref{noisymubs}) that for $\lambda,\alpha>0$ we have $\P_{\alpha}\in \mathcal C^{\rm sym}_{C}$ if and only if $g_d(\alpha)\geq \lambda$ where $g_d(\alpha) = \tfrac 1d \left((d-2)(1-\alpha) + 2\sqrt{1-\alpha}\sqrt{1+(d-1)\alpha}\right)$, and we now note that the same argument clearly applies also for $\alpha<0$. If $\alpha>0$ we can write this equivalently as $\alpha\leq g_d(\lambda)$, the function $\alpha\mapsto g_d(\alpha)$ is decreasing for $\alpha\in [0,1]$ and is its own inverse. If $\alpha<0$ the result still holds, but the inequality cannot be inverted using $g_d$, as $\alpha\mapsto g_d(\alpha)$ is increasing for $\alpha \in [-(d-1)^{-1},0]$ with inverse
$u\mapsto \tfrac 1d \left((d-2)(1-u) - 2\sqrt{1-u}\sqrt{1+(d-1)u}\right)$.
To summarise the $\lambda>0$ case:
$$
\mathcal C^{\rm sym}_{C}=\{\P_\alpha^{C}\mid g_d(\alpha)\geq \lambda\} \quad \text{when }\lambda\in [0,1].
$$
Now if $\lambda<0$, the corollaries Cor.\ \ref{cohcor} and Cor.\ \ref{cohcor2} do not completely determine $\mathcal C_C^{\rm sym}$. Indeed, Cor.\ \ref{cohcor} gives the necessary condition
$\lambda \geq -g_d(\alpha)$ for $\P_\alpha^C\in \mathcal C_C^{\rm sym}$, and Cor.\ \ref{cohcor2} the sufficient condition $\lambda \geq -g_d(\alpha)(d-1)^{-1}$, 
which only coincide in the qubit case. However, since $C * \Q_0$ and $\P_\alpha$ have the exact same form, we can interchange $\alpha$ and $\lambda$ above to conclude that for $\alpha>0$ and $\lambda\in [-(d-1)^{-1},1]$ we have $\P_{\alpha}\in \mathcal C^{\rm sym}_{C}$ if and only if $g_d(\lambda)\geq \alpha$. This already gives $\alpha_C=g_d(\lambda)$ in Eq. \eqref{CI}. In order to fully characterise $\mathcal C_C^{\rm sym}$ we need to show that $\P_\alpha^C\in \mathcal C_C^{\rm sym}$ for all $(\alpha,\lambda)\in [-(d-1)^{-1},0]\times [-(d-1)^{-1},0]$; this then gives
\begin{align*}
\mathcal C^{\rm sym}_{C}&=\{\P_\alpha^{C}\mid -(d-1)^{-1}\leq \alpha \leq g_d(\lambda)\},\\
&\text{when }\lambda\in [ -(d-1)^{-1}, 0].
\end{align*}
To prove the remaining bit it suffices (by convexity) to show that $\P_\alpha\in \mathcal C_C^{\rm sym}$ for the ``corner'' $\alpha=\lambda = -(d-1)^{-1}$, which was done in the main text.

\section{Example -- centrosymmetric case in dimension $3$}\label{A:centrod3}

Here we give a nontrivial example of the theory developed in the main text (and the Appendixes above). This example is relevant for the $N=2$ case of the spin-boson model but we work it out slightly more generally.

Let $C$ be any $3\times 3$ GIO matrix with real positive entries such that $(13)\in G_C$; that is, the symmetry group contains the permutation which exchanges 1 and 3 and leaves 2 unchanged. Then $C$ must be \emph{centrosymmetric}, i.e.,\ (also) symmetric about the counter-diagonal, so
\begin{align*}
C &= \begin{pmatrix}1 & \lambda & \gamma\\
\lambda & 1 & \lambda\\
\gamma & \lambda & 1\end{pmatrix},
\end{align*}
for some $\lambda,\gamma\in [0,1]$ and $D:=\tfrac 12(1+\gamma)-\lambda^2\geq 0$. The conditions ensure that $C\geq 0$. This covers both the uniform coherence in dimension $3$ ($\gamma=\lambda$ with $G_C=S_3$), and the reduction $\tilde C$ of the spin-boson model for $N=2$ ($\gamma=\lambda^4$ with $G_C=\{e,(13)\}$). In the former case $D = (\lambda+\tfrac 12)(1-\lambda)$, and in the latter case $D=\tfrac 12(1-\lambda^2)^2$, which are indeed both positive for all $\lambda\in [0,1]$.

If $G_C=S_3$ (i.e.,\ $\gamma=\lambda$) we know from the main text that $\mathcal C^{\rm sym}_C=\{\P_\alpha^C\mid \lambda \leq g_3(\alpha)\}$, i.e.,\ has affine dimension one. We now proceed to characterise $\mathcal C^{\rm sym}_{C}$ assuming $G_C=\{e, (13)\}$ (i.e., $\gamma\neq \lambda$). Denote $\pi_0 = (13)$ (as in the main text). We first note that each $\{e,(13)\}$-covariant incoherent observable $\P$ has $\P(2)=U_{\pi_0}^{\dagger}\P(0)U_{\pi_0}$ and $\P(1)=U_{\pi_0}^{\dagger}\P(1)U_{\pi_0}$. Therefore, it is of the form $\P=\P_{\bf q}$ for some ${\bf q} =(q,p,r)\in  \mathcal M:=\Delta\times [0,\tfrac 12]$ where $\Delta :=\{(q,p)\in [0,1]^2\mid p+q\leq 1\}$, and
\begin{align*}
\P_{\bf q}(0) &=\begin{pmatrix}p & 0& 0\\
0& r & 0\\
0&0& q
\end{pmatrix} &
\P_{\bf q}(1) &=\begin{pmatrix}s & 0& 0\\
0& 1-2r & 0\\
0&0& s 
\end{pmatrix} \\
\P_{\bf q}(2) &=\begin{pmatrix} q & 0& 0\\
0& r & 0\\
0&0& p
\end{pmatrix},
\end{align*}
with $s=1-p-q$. Since the map ${\bf q}\mapsto \P_{\bf q}$ is convex, the convex structure of the set of incoherent observables (including the shape of $\mathcal C^{\rm sym}_{C}$ inside it) is faithfully represented inside $\mathcal M$. In particular, the incoherent basis observable $\P_0(j) = |j\rangle\langle j|$ and its permutation $\P_0(\pi_0(j))$ are represented by the extremal points $(1,0,0)$ and $(0,1,0)$, while the trivial observables $\P(j) =\mu(j)\id$ (where $\mu$ is $\pi$-invariant) form the line from the origin (with $\P(1)=\id$) to $\tfrac 12(1,1,1)$ (with $\P(0)=\P(2) =\tfrac 12\id$). In particular, the centroid $\tfrac 14(1,1,1)$ of $\mathcal M$ is the ``coin toss'' observable $\P(j) = \binom{2}{j}\tfrac 14\id$, while the uniform trivial observable $\P(j) =\tfrac 13 \id$ is $\tfrac 13 (1,1,1)$. The former appears in the spin-boson model as the depolarisation of the spectral measure of the Hamiltonian (see the main text).

In order to state the result, we define the functions $w_+:\Delta \to [0,2]$, $w_-:\Delta\to [0,1]$, $w_0:\Delta\to [\gamma-1,\gamma+1]$, and
$w^-_{0}:\Delta\to [0,1+\gamma]$ by
\begin{align*}
w_{\pm}(p,q) &=(\sqrt{q}\pm \sqrt{p})^2,\quad w_0(p,q) =\gamma-1+2(q+p),\\
w_{0}^-(p,q) &= \begin{cases} w_-, & w_+\leq 1-\gamma\\ w_0, & w_+\geq 1-\gamma \end{cases}.
\end{align*}
Clearly, $w_-\leq w_+$ (for any $p,q$). Moreover, $w_0\leq w_-$ when $w_+\leq 1-\gamma$, with $w_0=w_-$ when $w_+=1-\gamma$. Correspondingly, $w_0\leq w_+$ when $w_-\leq 1-\gamma$, with $w_0=w_+$ when $w_-=1-\gamma$. In particular, $0\leq w_0^-\leq w_+$, and $w_0^-$ is a continuous function.

Since $D\geq 0$ we have $0\leq 2D\leq 1+\gamma$ and $2\lambda^2\leq 1+\gamma$. Therefore, we can define the functions
\begin{align*}
h_-:&[0,1+\gamma]\to [0,\lambda^2/(1+\gamma)],\\
h_{-}(w)&= \begin{cases}0, & w\in [0, 2D]\\
\big(\tfrac{\lambda \sqrt{w} -\sqrt{1+\gamma -w}\sqrt{D}}{1+\gamma}\big)^2, & w\in [2D,1+\gamma]
\end{cases}\\
h_+:&[0,1+\gamma]\to [D/(1+\gamma),\tfrac 12],\\
h_+(w) &= \begin{cases} 
\big(\tfrac{\lambda \sqrt{w} +\sqrt{1+\gamma -w}\sqrt{D}}{1+\gamma}\big)^2, & w\in [0,2\lambda^2]\\
\frac 12 & w\in [2\lambda^2,2].
\end{cases}
\end{align*}
One can readily check that these functions are continuous. The following result characterises $\mathcal C^{\rm sym}_C$ explicitly:
\begin{appendixproposition}\label{Cprop}
$$\mathcal C^{\rm sym}_{C} = \{ {\bf q}\in \mathcal M\mid w_-\leq 1-\gamma, \,r\in [h_-(w_0^-),h_+(w_+)]\}.$$
\end{appendixproposition}

Before giving a proof, we apply Prop.\ \ref{Cprop} to a convex line $\P_\alpha$ of the form $\P_\alpha(j) = \alpha |j\rangle\langle j| +(1-\alpha) \mu(j)\id$, where $\mu=(\mu(0),\mu(1), \mu(2))$ is a probability distribution, which must satisfy $\mu(0)=\mu(2)=t\in [0,\tfrac 12]$ and $\mu(1)=1-2t$ for $\P_\alpha$ to be $\{e,(13)\}$-covariant. Fixing $t$ we then have $\P_\alpha$ represented by the line $$p=\alpha + (1-\alpha)t, \quad q=r=(1-\alpha)t,$$
inside $\mathcal M$. The goal is to find to value of $\alpha$ at which it intersects the boundary of $\mathcal C^{\rm sym}_C$. The reason for not restricting to the canonical line $\P^C_\alpha$ (i.e.,\ $t=\tfrac 13$) is that we can cover also the cases where $C$ is obtained as a reduction from some higher dimension as described in Appendix \ref{A:symmetry}. In particular, the case of the spin-boson model for $N=2$ corresponds to $t=\tfrac 14$.

Since $r$ decreases as $\alpha$ increases, the intersection point must lie on the lower boundary surface $h_-(w_0^-)$. We restrict to the case $t\geq \tfrac 14(1-\gamma)$ for simplicity, because then $w_+(t,t) = 4t \geq 1-\gamma$, and hence $w_+\geq 1-\gamma$ on the whole line. Therefore $w_0^-=w_0=\gamma+1-2(1-\alpha)(1-2t)$ on the line, and so $\P^C_\alpha\in \mathcal C_C$ iff $h_-(w_0)\leq (1-\alpha)t$, which reads
$$
\lambda\sqrt{w_0}-\sqrt{1+\gamma-w_0} \sqrt D \leq (1+\gamma)\sqrt{(1-\alpha)t}.
$$ 
Rearranging this yields $$\lambda^2w_0\leq (1-\alpha)(\sqrt{2(1-2t)}\sqrt D + (1+\gamma)\sqrt{t})^2,$$ from which one can conveniently solve $\alpha$ as
$$
\alpha \leq 1- \frac{\lambda^2}{1-t(1-\gamma) +2\sqrt{2t(1-2t)D}}.
$$
In particular, for $t=\tfrac 14$ (the centroid of $\mathcal M$ corresponding to $\P(j) = \binom{2}{j}4^{-1}\id$) we get
\begin{equation}\label{centroid}
\alpha \leq 1- \frac{4\lambda^2}{3 +\gamma +4\sqrt{D}},
\end{equation}
while the case $t=\tfrac 13$ (the uniform trivial observable) gives instead
$$
\alpha \leq 1-\frac{3\lambda^2}{2+\gamma +2\sqrt{2D}},
$$
which for uniform decoherence, $\gamma = \lambda$, reduces to
$$
\alpha \leq 1-\frac{3\lambda^2}{2+\lambda +2\sqrt{(1+2\lambda)(1-\lambda)}}.
$$
One can easily check that the right-hand side is equal to $g_3(\lambda)$ appearing in the main text. Since $\alpha \leq g_3(\lambda)$ is equivalent to $\lambda \leq g_3(\alpha)$ (as $\alpha,\lambda>0$), the results are consistent.

{\it Proof of Prop.\ \ref{Cprop}}
As per the reduction method in the main text (proved in Appendix \ref{A:symmetry}), $P_{\bf q}\in \mathcal C^{\rm sym}_{C}$ if and only if there exist real matrices $C(j)$ satisfying Eqs. \eqref{ADcovC} and \eqref{GII0}. This forces the matrices to have the following form, where $a,b,c\in \mathbb R$:
\begin{align*}
C(0) &=\begin{pmatrix}p & a& c\\
a& r & b\\
c&b& q
\end{pmatrix}, \\
C(1) &=\begin{pmatrix}s & \lambda-a-b& \gamma-2c\\
\lambda-a-b& 1-2r & \lambda-a-b\\
\gamma-2c&\lambda-a-b& s 
\end{pmatrix},\\
C(2) &=\begin{pmatrix} q & b& c\\
b& r & a\\
c&a& p 
\end{pmatrix}. 
\end{align*}
The problem is, then, whether we can find $a,b,c\in \mathbb R$ so that \emph{the first two} of these three matrices are positive semidefinite. (Note that $C(2)=U_{\pi_0} C(0)U^{\dagger}_{\pi_0}$ is then automatically positive semidefinite.)

In order to further simplify the positivity condition for $C(1)$ (which is easier of the two), we note that all \emph{centrosymmetric} matrices (i.e.,\ ones commuting with $U_{\pi_0}$) can be brought to a block form by a specific orthogonal matrix $Q$ only depending on $U_{\pi_0}$ \cite{liu03}; in our case,
$$Q =\frac 1{\sqrt 2}\begin{pmatrix}1 & 0 & 1\\
0 & \sqrt{2} & 0\\
-1 & 0 & 1
\end{pmatrix},$$
and letting $x = a+b$, $y=a-b$ we obtain
\begin{align*}
Q^TC(0)Q &= \begin{pmatrix} \frac{p+q}{2}-c & \frac{y}{\sqrt 2} & \frac{p-q}2\\
\frac{y}{\sqrt 2} & r& \frac{x}{\sqrt 2}\\
\frac{p-q} 2 & \frac{x}{\sqrt 2} & \frac{p+q}{2}+c 
\end{pmatrix},\\
Q^TC(1)Q &= \begin{pmatrix} s-\gamma+2c & 0 & 0\\
0 & 1-2r& \sqrt{2}(\lambda-x)\\
0 & \sqrt{2}(\lambda-x) & s+\gamma-2c 
\end{pmatrix}.
\end{align*}
Since orthogonal transformations preserve positivity (and determinants), we can extract the conditions for $C(j)\geq 0$ from these matrices. First of all, $\tfrac{p+q}{2}+c\geq 0$ is clearly necessary for $C(0)\geq 0$. The $2\times 2$ principal minors of $Q^TC(0)Q$ are
\begin{align*}
d_{(0)} &:=r(\tfrac{p+q}{2}+c)-\tfrac {x^2}{2},\\
d_{(1)} &:= (\tfrac {p+q}{2}-c)(\tfrac {p+q}{2}+c)-(\tfrac{p-q}{2})^2 = qp-c^2,\\
d_{(2)} &:= r(\tfrac{p+q}{2}-c)-\tfrac {y^2}{2}.
\end{align*}
Assuming $\tfrac {p+q}{2}+c> 0$, we can write
\begin{align*}
\det C(0) &= r d_{(1)} - \tfrac 12(\tfrac {p+q}{2}-c)x^2 -\tfrac 12 (\tfrac {p+q}{2}+c)y^2+xy\tfrac{p-q}2\\
&= d_{(0)}d_{(1)} / (\tfrac {p+q}{2}+c)- \tfrac 12(\tfrac {p+q}{2}+c)(y-y_0)^2,
\end{align*}
where $y_0=\tfrac{1}{2} x(p-q)/ (\tfrac {p+q}{2}+c)$. Hence, $C(0)\geq 0$ is equivalent to $\det C(0)\geq 0$ and $d_{(i)}\geq 0$ for $i=1,2,3$. Since $d_{(0)},d_{(1)}$ do not depend on $y$ and imply $d_{(2)}\geq 0$ when $y=y_0$, it follows that if $C(0)\geq 0$ for some choices of $x,y,c$, it also holds if we take $y=y_0$ (as the determinant can only  increase). As $C(1)$ does not depend on $y$, we may therefore always take $y=y_0$. With this choice, $C(0)\geq 0$ if and only if
\begin{align}\label{C0}
 \tfrac{p+q}{2}+c&\geq 0, &|x| &\leq \sqrt{2r (\tfrac{p+q}{2}+c)}, & c^2&\leq qp.
\end{align}
(In the special case $\tfrac{p+q}{2}+c=0$, we have $C(0)\geq0$ only if $p=q$ and $x=0$, so $C(0)\geq 0$ iff $y^2\leq 4rp$. Hence we can take $y=y_0=0$, and this case is covered by \eqref{C0}.)

Next we observe that $C(1)\geq 0$ if and only if
\begin{align}\label{C1}
\left|c -\frac{\gamma}{2}\right|&\leq \frac{s}{2}, &
|\lambda-x| &\leq \sqrt{(1-2r)(\tfrac{s +\gamma}2-c)}.
\end{align}
For fixed $c$, the inequalities \eqref{C0} and \eqref{C1} force $x$ into an intersection of two intervals. By the triangle inequality, \eqref{C0} and \eqref{C1} hold for \emph{some} $x$, if and only if $\lambda\leq g(c)$ where
$$
g(c):= \sqrt{2r (\tfrac{p+q}{2}+c)}+\sqrt{(1-2r)(\tfrac{s +\gamma}2-c)},
$$
and the remaining constraints hold for $c$. These constraints are given by the following set:
\begin{align*}
\mathcal D &:= \left\{c\in \mathbb R\,\Big|\, \left|c-\tfrac{\gamma}{2}\right| \leq \tfrac{s}{2}, \, |c|\leq \sqrt{pq}\right\}.
\end{align*}
Hence, ${\bf q}\in \mathcal C_{C}$ if and only if $\lambda \leq g(c)$ for some $c\in \mathcal D$. Clearly, this is in turn equivalent to the following:
\begin{equation}\label{eqf}
\mathcal D\neq \emptyset \text{ and }\lambda \leq \max_{c\in \mathcal D}g(c).
\end{equation}
By the triangle inequality, $\mathcal D\neq \emptyset$ if and only if 
\begin{equation}\label{hellinger02}
\tfrac{\gamma-s}{2} \leq \sqrt{qp},
\end{equation}
in which case $\mathcal D$ is the interval
\begin{equation}\label{domain}
\mathcal D :=\left[\max\{-\sqrt{qp},\tfrac{\gamma-s}2\}, \min\left\{\sqrt{qp},\tfrac{\gamma+s}2 \right\}\right]
\end{equation}
(where the left boundary does not exceed the right). Hence, \eqref{hellinger02} is a \emph{necessary} (but not sufficient) condition for ${\bf q}\in \mathcal C_{C}$. In fact, it is one of the two Hellinger distance conditions given by Cor.\ 1 of the main text. We note that \eqref{hellinger02} does not depend on $r$, and let $\mathcal R$ denote the set of those $(p,q)\in \Delta$ for which it holds. Then
$$\mathcal R =\{(p,q)\in \Delta\mid w_-\leq 1-\gamma\},$$
where the function $w_-=w_-(p,q)$ was defined above. Next we note that the maximal domain of $g$ (where the square roots are defined) is $$\mathcal D_{\rm max}:=\left[-\tfrac12(p+q),\tfrac12(s+\gamma)\right],$$
which clearly contains $\mathcal D$ because $\sqrt{qp}\leq \tfrac 12(q+p)$. We readily find the global maximum point of $g$ within $\mathcal D_{\rm max}$: 
$$
2\frac{dg}{dc} = \sqrt{\tfrac{2r}{c+(p+q)/2}}-\sqrt{\tfrac{1-2r}{(s+\gamma)/2-c}},
$$
and hence $\frac{dg}{dc}\geq 0$ iff
$c\leq c_0$ where $c_0 = r(1+\gamma)-\tfrac12 (p+q)$.
Note that indeed $c_0\in \mathcal D_{\rm max}$, as
\begin{align*}
c_0+\tfrac{p+q}{2} &= r(1+\gamma)\geq 0, & \tfrac{s+\gamma}2 -c_0&=(\tfrac12-r)(1+\gamma)\geq 0.
\end{align*}
It follows that $g(c_0) = \sqrt{\tfrac12(1+\gamma)}$, so that $g(c_0)\geq \lambda$ automatically by the positivity of $C$, and hence $c_0\in \mathcal D$ is a \emph{sufficient} condition for ${\bf q}\in \mathcal C_{C}$. However, it is not a necessary condition, as $c_0$ may fall on either side of $\mathcal D$; in those cases, the maximum is attained at the boundary, and we obtain a constraint in terms of $\lambda$. 
In order to find it we consider the three possible cases for $\mathcal D$:
\begin{enumerate}
\item $\mathcal D_1 = [-\sqrt{qp}, \sqrt{qp}]$.
\item $\mathcal D_2 = [\tfrac{\gamma-s}{2}, \sqrt{qp}]$.
\item $\mathcal D_3 = [\tfrac{\gamma-s}{2},\tfrac{\gamma+s}{2}]$.
\end{enumerate}
(The case $\mathcal D =[-\sqrt{qp},\tfrac{\gamma+s}{2}]$ cannot occur as $\gamma\geq 0$.) Since $\mathcal D$ does not depend on $r$, these cases set up a unique partition of $\mathcal R$; we have $\mathcal R = \mathcal R_1\cup \mathcal R_2\cup\mathcal R_3$ where $\mathcal R_i:=\{(p,q)\in \mathcal R\mid \mathcal D=\mathcal D_i\}$. It is then easy to check that
\begin{align*}
\mathcal R_1 &=\{(p,q)\in \Delta\mid w_+ < 1-\gamma\},\\
\mathcal R_2 &= \{(p,q)\in \Delta\mid 1-\gamma\leq w_+ < 1+\gamma\}\cap \mathcal R,\\
\mathcal R_3 &= \{(p,q)\in \Delta \mid w_+\geq 1+\gamma\},
\end{align*}
where we have used the function $w_+=w_+(p,q)$ defined above.
Each of these cases then has three subcases (a)--(c) according to whether $c_0\leq \min \mathcal D$, $c_0\in \mathcal D$, or $c_0\geq \max  \mathcal D$, respectively. In order to describe them we also need the functions $w_0=w_0(p,q)$, and $l(w) :=\tfrac 12w/(1+\gamma)$. We observe that ${\bf q}\in \mathcal C_{C}$ if and only if $(q,p)$ falls into one of the following eight categories:
\begin{enumerate}
\item $(p,q)\in \mathcal R_1$ (implying $0\leq l(w_-)\leq l(w_+)\leq \tfrac 12$), and
\begin{enumerate}
\item $0\leq r\leq l(w_-)$ and $\lambda \leq g(-\sqrt{qp})$, or
\item $l(w_-)\leq r\leq l(w_+)$, or
\item $l(w_+)\leq r\leq \tfrac 12$ and $\lambda\leq g(\sqrt{qp})$.
\end{enumerate}
\item $(p,q)\in \mathcal R_2$ (implying $0\leq l(w_0)\leq l(w_+)\leq \tfrac 12$) and
\begin{enumerate}
\item $0\leq r\leq l(w_0)$ and $\lambda\leq g(\tfrac{\gamma-s}{2})$, or
\item $l(w_0)\leq r\leq l(w_+)$ or 
\item $l(w_+)\leq r\leq \tfrac 12$ and $\lambda\leq g(\sqrt{qp})$.
\end{enumerate}
\item $(p,q)\in \mathcal R_3$ (implying $0\leq l(w_0)\leq \tfrac 12$), and
\begin{enumerate}
\item $0\leq r\leq l(w_0)$ and $\lambda\leq g(\tfrac{\gamma-s}{2})$, or
\item $l(w_0)\leq r\leq \tfrac 12$.
\end{enumerate}
(In this case (c) does not occur.)
\end{enumerate}
We then notice that
\begin{align}
g(\pm \sqrt{qp}) &= \sqrt{rw_\pm} +\sqrt{(\tfrac 12-r)(1+\gamma - w_\pm)},\nonumber\\
g(\tfrac 12(\gamma-s)) &= \sqrt{rw_0} +\sqrt{(\tfrac 12-r)(1+\gamma - w_0)};\label{gcond}
\end{align}
hence in each case the condition involving $\lambda$ is of the form
\begin{equation}\label{sqcond}
\lambda \leq \sqrt{rw} +\sqrt{(\tfrac 12-r)(1+\gamma - w)},
\end{equation}
for $(r,w)\in [0,\tfrac 12]\times [0,1+\gamma]$. Equivalently,
\begin{equation}\label{sqcond2}
h_-(w) \leq r \leq h_+(w),
\end{equation}
where $h_\pm$ were introduced above. We observe that
\begin{equation}\label{bounds}
h_{-}(w)\leq l(w)\leq h_+(w)
\end{equation}
for all $(r,w)\in [0,\tfrac 12]\times [0,1+\gamma]$. (In fact, the region defined by \eqref{sqcond} is symmetric about the line $r=l(w)$.) Using \eqref{gcond}, \eqref{sqcond2}, and \eqref{bounds} we can put together subcases (a)--(c) in the above three cases: ${\bf q}\in \mathcal C_{C}$ if and only if one of the following conditions hold:
\begin{enumerate}
\item $(p,q)\in \mathcal R_1$ and $h_-(w_-)\leq r \leq h_+(w_+)$.
\item $(p,q)\in \mathcal R_2$ and $h_-(w_0)\leq r \leq h_+(w_+)$.
\item $(p,q)\in \mathcal R_3$ and $h_-(w_0)\leq r \leq \tfrac 12$.
\end{enumerate}
Noting that $w_+\geq 1+\gamma\geq 2\lambda^2$ when $(p,q)\in \mathcal R_3$, we have $h_+(w_+)=\tfrac 12$ when $(p,q)\in \mathcal R_3$, and hence the upper bound for $r$ is always $h_+(w_+)$. By the definition of $w_0^-$, the lower bound is $h_-(w_0^-)$, and the proof of the proposition is complete.

\section{Spin-boson model with $N=2$}\label{A:sbN2}

Here we present in detail the analytical solution of $\mathcal C_C^{\rm sym}$ for $N=2$ in the spin-boson model. We order the two-qubit incoherent basis in the usual way as $\{|00\rangle, |01\rangle, |10\rangle, |11\rangle\}$, with respective label set $\{1,2,3,4\}$, on which the symmetric group $S_4$ acts. In this basis our matrix $C[\lambda]$ of the dynamical GIO reads
\begin{align*}
C[\lambda]=\begin{pmatrix}
1 & \lambda & \lambda & \lambda^4\\
\lambda & 1 & 1 & \lambda\\
\lambda & 1 & 1 & \lambda\\
\lambda^4 & \lambda & \lambda & 1
\end{pmatrix}\,.
\end{align*}
The equivalence classes of maximal coherence are given by $I^C_0=\{1\}$, $I^C_1=\{2,3\}$, $I^C_2=\{4\}$, so $n_C=3$. These correspond to the eigenspaces of $S_z$ given by $\mathcal S_0 = {\rm span}\{|00\rangle\}$, $\mathcal S_1 = {\rm span}\{|01\rangle, |10\rangle\}$, and $\mathcal S_2 = {\rm span}\{|11\rangle\}$. Note that $\mathcal S_1$ is a nontrivial two-dimensional decoherence-free subspace, as we can see from the matrix. Now the canonical incoherent observable $\P^C$ has three outcomes $\{0,1,2\}$, and is given by
$\P^C(0)= |00\rangle\langle 00|$, $\P^C(1)= |10\rangle\langle 01|+|10\rangle\langle 01|$, $\P^C(2)=|11\rangle\langle 11|$, which is just the spectral decomposition of $S_z$ as mentioned in the main text. The line $\P^C_\alpha$ used to define the quantity $\alpha_C$ is given by $\P^C_\alpha(j) = \alpha \P^C(j) +(1-\alpha) {\rm tr}[\P^C(j)]\tfrac 14\id= \alpha \P^C(j) +(1-\alpha) \binom{2}{j}\tfrac 14\id$. Finally, the symmetry group $G_C$ leaving $C$ unchanged is the subgroup of $S_4$ generated by the within-class permutation $(23)$ (exchanging $2$ and $3$), and the order-reversal $(14)(23)$ , that is, 
$G_C=\{e, (23), (14)(23),(14)\}$. The task is to characterise the set $\mathcal C_C^{\rm sym}$.

We now carry out the reduction to dimension $n_C=3$. First, we have 
\begin{align*}
\tilde C[\lambda]=\begin{pmatrix}
1 & \lambda & \lambda^4\\
\lambda & 1 & \lambda\\
\lambda^4 & \lambda & 1
\end{pmatrix},
\end{align*}
with $G_{\tilde C}= \{e, (13)\}$. The homomorphism $\phi:G_C\to G_{\tilde C}$ defined by $\pi(I^C_k)=I^C_{\phi(\pi)(k)}$ maps as follows: $\phi((23)) =\phi(e) = e$, $\phi((14))=\phi((14)(23))=(13)$, so that $\phi(G_C) = \{e, (13)\}=G_{\tilde C}$, i.e. the within-class permutation is mapped to the identity, and the reversal carries over to the reduction. In this way we end up with the case considered above in Appendix \ref{A:centrod3} with $\gamma=\lambda^4$, so Prop.\ \ref{Cprop} gives the three-dimensional convex set $\mathcal C_C[\phi(G_C)]=\mathcal C_{C}^{\rm sym}$. Hence every $\P\in \mathcal C_C^{\rm sym}$ is of the form
\begin{align*}
\P(0) &=\begin{pmatrix}p & 0& 0 & 0\\
0& r & 0 & 0\\
0& 0 & r & 0\\
0&0& 0 & q
\end{pmatrix}, &
\P(1) &=\begin{pmatrix}s & 0& 0 & 0\\
0& 1-2r & 0 & 0\\
0& 0 & 1-2r & 0\\
0&0& 0 &  s 
\end{pmatrix}, \\
\P(2) &=\begin{pmatrix} q & 0& 0 & 0\\
0& r & 0&0\\
0& 0 & r&0\\
0&0& 0 & p
\end{pmatrix},
\end{align*}
where $(p,q,r)\in [0,1]$, $r\in [0,\tfrac 12]$ are such that $w_-(p,q)\leq 1-\lambda^4$, $r\in [h_-(w_0^-(p,q)),h_+(w_+(p,q))]$ (see Appendix \ref{A:centrod3}). As noted in the main text, $\mathcal C_C^{\rm sym}$ is therefore a convex set of affine dimension $3$ only depending on $\lambda$, and can conveniently be plotted in the parameterisation $(p,q,r)$, as shown in Fig.\ \ref{fig:spinboson1} in the main text. Furthermore, by substituting $\gamma =\lambda^4$ into Eq. \eqref{centroid} in Appendix \ref{A:centrod3} we immediately deduce that
$$
\alpha_{2}(\lambda) = 1-\frac{4\lambda^2}{3+\lambda^4+2\sqrt2(1-\lambda^2)},
$$
as claimed in the main text.

\section{Small coherence limit in the spin-boson model}\label{A:smallcoh}

Here we prove the asymptotic behaviour of the curve $\alpha = \alpha_N(\lambda)$ stated in Prop.\ \ref{asymptotic}. We first consider the upper bound $\alpha = U_N(\lambda)$, or, equivalently, $\lambda = \beta_{01}(\alpha)$, which is an explicit algebraic curve given by the Hellinger distance $d^2_{01}(\P^{\tilde C}_\alpha)$ corresponding to the dominant coherence $\lambda$ in $C[\lambda]$, as explained in the main text. Hence the asymptotic form is easily obtained: $\beta_{01}(\alpha) = k_N \sqrt{1-\alpha} + O(1-\alpha)$, where $k_N = 2^{-N/2}(1+\sqrt{N})$. Here we have used the customary notation where $g(\alpha) = O((1-\alpha)^k)$ means that the function $\alpha \mapsto |g(\alpha)|/(1-\alpha)^k$ is bounded on some neighbourhood of $\alpha=1$. Note that this does not require a convergent series expansion for $g$ at $\alpha=1$. Indeed, while such an expansion exists for $\beta_{01}(\alpha)$, the same is not clear for the exact curve $\alpha = \alpha_N(\lambda)$, which nevertheless turns out to have the same behaviour. In order to see this we find a lower bound with the same asymptotic behaviour. We denote $q_k(\alpha) = \binom{N}{k} \frac{1-\alpha}{2^N}$, for each $k=0,\ldots,N$, and $u_k(\alpha) = \sqrt{q_k(\alpha)(\alpha+q_k(\alpha))}$ for $k=0,1$. Then the upper bound reads $\beta_{01}(\alpha) = u_0(\alpha) + u_1(\alpha) +(1-\alpha)(1-\tfrac{1}{2^N}(1+N))$. It turns out that the first two terms form a lower bound:
\begin{appendixlemma}\label{lemG} $u_0(\alpha)+u_1(\alpha)\leq \alpha_N^{-1}(\alpha)$ for all $\alpha\in [0,1]$, where $\alpha_N^{-1}:[0,1]\to [0,1]$ is the inverse of the monotone function $\alpha_N$.
\end{appendixlemma}
As a consequence, we obtain the following result, the first part of which is Prop.\ \ref{asymptotic} in the main text:
\begin{appendixproposition}\label{asympt} For any fixed $N$,
$$
\alpha_N^{-1}(\alpha) = k_N \sqrt{1-\alpha} + O(1-\alpha) \,\, \text{ as } \alpha\rightarrow 1.
$$
For each $\alpha\in [0,1]$ and $N$, we have the error bound
$$
|\alpha_N^{-1}(\alpha)-k_N \sqrt{1-\alpha}| \leq 1-\alpha +\tfrac 32 k_N \left(\tfrac{1-\alpha}{\alpha}\right)^{\frac 32}.
$$
\end{appendixproposition}
\begin{proof}
Using the bijection $\alpha\mapsto v= \sqrt{1-\alpha}$ we define
\begin{align*}
f(v) &:= u_0(\alpha) +u_1(\alpha) \\
&= v\sqrt{1-r}\sqrt{1-rv^2}+v\sqrt{1-s}\sqrt{1-sv^2},
\end{align*}
where $r=1-2^{-N}$ and $s=1-N2^{-N}$. Now
$\beta_{01}(\alpha) = f(v) + v^2(1-\tfrac{1}{2^N}(1+N))\leq f(v) +v^2$. Combined with Lemma \ref{lemG}, this gives $f(v) \leq \alpha_N^{-1}(\alpha)\leq f(v) + v^2$, that is, $|\alpha_N^{-1}(\alpha)-f(v)| \leq v^2$ for all $v\in [0,1]$. Since $f(0)=0$, Taylor's theorem gives
$f(v) = k_N v +z(v)$, where $|z(v)| \leq  \frac 12v^2\max_{0\leq \tilde v \leq v} |\tfrac{d^2f}{dv^2}(\tilde v)|$. Now 
\begin{align*}
\left|\frac{d^2f}{dv^2}(\tilde v)\right| &= \frac{r\sqrt{1-r}(3-2r\tilde v^2)\tilde v}{(1-r\tilde v^2)^{\frac 32}}+\frac{s\sqrt{1-s} (3-2s\tilde v^2)\tilde v}{(1-s\tilde v^2)^{\frac 32}}\\
&\leq \frac{3vr\sqrt{1-r} }{(1-rv^2)^{\frac 32}}+\frac{3vs\sqrt{1-s} }{(1-sv^2)^{\frac 32}} \leq \frac{3(1+\sqrt N)v}{2^{N/2}(1-v^2)^{\frac 32}},
\end{align*}
and hence $|z(v)|\leq \tfrac 32 k_N v^3(1-v^2)^{-\frac 32}$. Therefore,
\begin{align*}
|\alpha_N^{-1}(\alpha)-k_N v| &=|\alpha_N^{-1}(\alpha)- f(v)+z(v)|\\
&\leq |\alpha_N^{-1}(\alpha)- f(v)|+|z(v)|\\
&\leq v^2 +\tfrac 32 k_N v^3(1-v^2)^{-\frac 32}.
\end{align*}
Substituting $v=\sqrt{1-\alpha}$ yields the claim.
\end{proof}
Since $\alpha_N^{-1}(1)=0$, the validity of the approximation can be quantified by the relative error $\epsilon = |\alpha_N^{-1}(\alpha)-k_N v|/\alpha_N^{-1}(\alpha)$. Assuming that the second term in the error bound of the proposition is negligible for this consideration, we get the maximum relative error $\epsilon_{\rm max} \approx (k_N/\sqrt{1-\alpha}-1)^{-1}$. This suggests that the asymptotic form becomes valid around $\alpha \approx 1-c k_N^2$ where $c$ is a constant. For instance, $\alpha \approx 1-(k_N/24)^2$ corresponds to a maximum error of 4--5\%.

We now prove Lemma \ref{lemG} by constructing explicitly a GII satisfying Eqs. \eqref{ADcovC} and \eqref{GII0}, for $\lambda = u_0(\alpha)+u_1(\alpha)$ and all $\alpha\in [0,1]$. (This then shows that $\P^{C}_\alpha\in \mathcal C_C$ whenever $\alpha$ satisfies $\lambda \leq u_0(\alpha)+u_1(\alpha)$, and therefore the transition point $\alpha=\alpha_N(\lambda)$ must satisfy $\lambda\geq u_0(\alpha)+u_1(\alpha)$, giving the claimed inequality.)
For simplicity, we show the construction for $N\geq 5$ (for smaller $N$ it needs a few modifications). 

First define
$$
C(0) = \left(\begin{array}{c|c|c}
\begin{matrix} \alpha+q_0 & u_0 & 0\\
u_0 & q_0 & 0\\
0 & 0 & q_0
\end{matrix}
& \mathbf{0} & \mathbf{0}\\
\hline
{\bf 0} & \begin{matrix} q_0 &  & \\
 & \ddots & \\
 &  & q_0
\end{matrix} & {\bf 0}\\
\hline
\mathbf{0} & \mathbf{0} & \begin{matrix} q_0 & 0 & 0\\
0 & q_0 & 0\\
0 & 0 & q_0
\end{matrix}
\end{array}\right)
$$
and
$$
C(1) = \left(\begin{array}{c|c|c}
\begin{matrix} q_1 & u_1 & q_1\\
u_1 & \alpha +q_1 & u_1\\
q_1 & u_1 & q_1
\end{matrix}
& \mathbf{0} & \mathbf{0}\\
\hline
{\bf 0} & \begin{matrix} q_1 &  & \\
 & \ddots & \\
 &  & q_1
\end{matrix} & {\bf 0}\\
\hline
\mathbf{0} & \mathbf{0} & \begin{matrix} q_1 & 0 & -q_1\\
0 & q_1 & 0\\
-q_1 & 0 & q_1
\end{matrix}
\end{array}\right).
$$
These matrices are clearly positive semidefinite. (If $N=5$ the central block is empty.) We then set $C(N) := U_{\pi_0}^\dagger C(0)U_{\pi_0}$, $C(N-1):=U_{\pi_0}^\dagger C(1)U_{\pi_0}$, where $U_{\pi_0}$ is the permutation matrix for the order-reversal $\pi_{0}$. Recalling that $A\mapsto U_{\pi_0}^\dagger AU_{\pi_0}$ transposes the matrix along the counter-diagonal, we can easily check that $C(0)+C(1)+C(N-1)+C(N)$ reads
$$
\left(\begin{array}{c|c|c}
\begin{matrix} \alpha+r_1 &\lambda & 0\\
\lambda & \alpha +r_1 & u_1\\
0 & u_1 & r_1
\end{matrix}
& \mathbf{0} & \mathbf{0}\\
\hline
{\bf 0} & \begin{matrix} r_1 &  & \\
 & \ddots & \\
 &  & r_1
\end{matrix} & {\bf 0}\\
\hline
\mathbf{0} & \mathbf{0} & \begin{matrix} r_1 & u_1 & 0\\
u_1 & \alpha+r_1 & \lambda\\
0 & \lambda & \alpha+r_1
\end{matrix}
\end{array}\right)
$$
where $r_k = 2\sum_{j=0}^k q_j$. Note that $-q_1$ on the lower right block of $C(1)$ cancels out the $q_1$ on the upper left block, and the overlap between $C(0)$ and $C(1)$ produces $\lambda$ on both ends of the diagonal $(n,n+1)$.

Next, let $j_0=N/2$ if $N$ is even, and $j_0=(N-1)/2$ if $N$ is odd, and define
\begin{align*}
M_{j}(w) &= \begin{pmatrix} q_j & w & 0\\
w & \alpha +q_j & w\\
0 & w & q_j
\end{pmatrix}
\end{align*}
for $j=2,\ldots, j_0$ and $w\in [0,1]$. Now $M_j(w)\geq 0$ if and only if $2w^2\leq q_j(q_j+\alpha)$, which is clearly the case for $w=u_0$ when $j\geq 2$, and for $w=u_1$ when $j\geq 3$, because $q_j$ is increasing in $j$ (up to $j_0$), and $q_2\geq 2q_0$, $q_3\geq 2q_1$. (Note that the same is not true for $j=0,1$, which is why we defined $C(0)$ and $C(1)$ differently above.) Then let
$$
V_j(w) = \begin{pmatrix} q_k\id_{j-1} & {\bf 0} & {\bf 0}\\ {\bf 0} & M_j(w) & {\bf 0}\\ {\bf 0} & {\bf 0} & q_k\id_{N-j-1}\end{pmatrix}
$$
for each $j=2,\ldots, j_0$.

We now consider even and odd $N$ separately.

If $N$ is even, $j_0$ is the ``middle'' point of $\{0,\ldots, N\}$, and we can set up $C(2), \ldots, C(j_0-1)$ as follows: $C(j):= V_j(u_1)$ if $j$ is odd, and $C(j) :=V_j(u_0)$ if $j$ is even. Notice that the block $M_j(w)$ moves down the diagonal as $j$ increases, with $w$ alternating between $u_1$ and $u_0$. These blocks only overlap between neighbouring matrices $C(j)$, $C(j+1)$, and hence the sum $\sum_j C(j)$ has $\lambda=u_0+u_1$ on each $(k,k-1)$-entry for $k=0,\ldots, j_0-1$. We then set $C(j_0+j) := U_{\pi_0}^\dagger C(j_0-j)U_{\pi_0}$ for $j=1,\ldots,j_0-2$, to satisfy the symmetry. Finally, we define the middle element by $C(j_0):=D + G$ where $D$ follows the above pattern, that is, $D:=V_{j_0}(u_1)$ ($V_{j_0}(u_0)$) if $j_0$ is odd (even), $G:= C[\lambda]-\tilde C$, and $\tilde C$ denotes the truncation of $C[\lambda]$ to second order in $\lambda$, so that $\tilde C$ is a tridiagonal matrix. Now $\sum_{j\neq j_0} C(j) +D = \tilde C$, and the remainder $G$ is included in $C(j_0)$ so that $\sum_j C(j) =C[\lambda]$. Note also that $C(j_0)$ satisfies $U_{\pi_0}^\dagger C(j_0)U_{\pi_0} = C(j_0)$ as $D$ has the block $M_{j_0}(w)$ exactly at the centre. We are left to prove that $C(j_0)\geq 0$. To do this we write
$$C(j_0) = D+ G = (D-\epsilon \id) + (\epsilon \id + G),$$ where we pick $\epsilon>0$ small enough so that $D-\epsilon \id\geq 0$, but large enough to make $\epsilon \id + G\geq 0$. We can take, for instance, $\epsilon = \tfrac 23 q_1$. In fact, first note that $\lambda=u_0+u_1\leq k_N\sqrt{1-\alpha}$, and $\lambda\leq 0.305$ for all $\alpha\in [0,1]$, $N\geq 5$. Also, $k_N^2\leq 2.1 N/2^N$ for $N\geq 5$ so
\begin{align*}
\frac 12\sum_{j\neq i} G_{ij} &\leq \sum_{k=2}^\infty \lambda^{k^2}\leq \lambda^4\sum_{k=0}^\infty \lambda^k= \frac{\lambda^4}{1-\lambda}\leq 0.14 \lambda^2\\
&\leq 0.14k_N^2 (1-\alpha) \leq 0.14*2.1 q_1\leq \frac 13 q_1,
\end{align*}
and hence $\epsilon \id + G\geq 0$ when $N\geq 5$ by diagonal dominance. Moreover,  $M_{j_0}(w)-\epsilon\id\geq 0$ since for $N\geq 6$,
\begin{align*}
&(q_{j_0}-\epsilon)(\alpha+q_{j_0}-\epsilon) -2w^2\\
&\geq (q_{3}-\epsilon)(\alpha+q_{3}-\epsilon) -2q_1(\alpha +q_1)\\
& \geq (q_3-2q_1-\epsilon)(\alpha+q_3)+2q_1(q_3-q_1)-q_3\epsilon\\
& = (q_3-\tfrac 83q_1)(\alpha+q_3)+\tfrac 43q_1(q_3-\tfrac 32q_1)\geq 0,
\end{align*}
where $q_3-\tfrac83 q_1 = \tfrac N6((N-1)(N-2)-16)\tfrac{1-\alpha}{2^N}\geq 0$ as $N\geq 6$. (Note that the $2\times2$ principal minors of $M_{j_0}(w)$ are then automatically positive.) Hence $D-\epsilon\id\geq 0$. This completes the construction for even $N$.

If $N$ is odd, we have two middle points $j_0$ and $j_0+1$. We now define $C(2), \ldots, C(j_0-1)$ as above, and again set $C(j_0+1+j) := U_{\pi_0}^\dagger C(j_0-j)U_{\pi_0}$, for $j=1,\ldots,j_0-2$. The sum of these matrices coincides with $\tilde C$ everywhere except in the $4\times 4$ block at the centre of the matrix, and on the main diagonal. The two remaining matrices have to be set up separately. Let $w=u_1$ if $j_0$ is odd, and $w=u_0$ if it is even. We first define $D$ as the $(N+1)\times (N+1)$ matrix having  $q_{j_0}$ on the main diagonal outside the central $4\times 4$ block,
\begin{equation}\label{eqB}
\begin{pmatrix}
q_{j_0} & w & 0 & 0\\
w & \alpha +q_{j_0} & \frac \lambda 2 & 0\\
0 & \frac \lambda 2 & q_{j_0} & 0\\
0 & 0& 0& q_{j_0}
\end{pmatrix},
\end{equation}
and the remaining elements zero. Now $D\geq 0$ due to $q_{j_0}\geq q_{2} \geq 2q_1$, which holds as $N\geq 5$. Clearly, the sum $D + U_{\pi_0}^\dagger DU_{\pi_0}$ has $2q_{j_0}=q_{j_0}+q_{j_0+1}$ on the main diagonal outside the central block, which reads
$$
\begin{pmatrix}
2q_{j_0} & w & 0 & 0\\
w & \alpha +2q_{j_0} & \lambda & 0\\
0 & \lambda & \alpha + 2q_{j_0} & w\\
0 & 0& w& 2q_{j_0}
\end{pmatrix}.
$$
Added to the previously constructed $C(j)$, this produces $\tilde C$. Now define $C(j_0):= D+\frac 12G$ and $C(j_0+1)=U_{\pi_0}^\dagger DU_{\pi_0}+\frac 12 G$, where $G$ is as before. Then $C(j_0+1) = U_{\pi_0}^\dagger C(j_0)U_{\pi_0}$ and $\sum_{j=0}^N C(j)=C[\lambda]$. As in the even case, we establish that $C(j_0) \geq 0$; we write $$C(j_0) = (D-\tfrac \epsilon2 \id) + \tfrac 12(\epsilon \id + G),$$ with the same $\epsilon$ as before, so the second term is positive for $N\geq 5$. Using $\tfrac 14\lambda^2 \leq \frac 12(u_0^2+u_1^2)$ we get
\begin{align*}
&(q_{j_0}-\tfrac \epsilon2)(\alpha+q_{j_0}-\tfrac \epsilon2) - (\lambda/2)^2-w^2\\
& \geq (q_{2}-\tfrac \epsilon2)(\alpha+q_{2}-\tfrac \epsilon2) - \tfrac 12(u_0^2+u_1^2)-q_1(\alpha +q_1)\\
& = (q_{2}-\tfrac 13q_1)(\alpha+q_{2}-\tfrac 13q_1) - \tfrac 12q_0(\alpha+q_0)-\tfrac 32q_1(\alpha +q_1)\\
& = (q_2-\tfrac{11}{6} q_1-\tfrac 12 q_0)\alpha +(q_2-\tfrac 13 q_1)^2-\tfrac 12 q_0^2-\tfrac 32 q_1^2\\
&= (\tfrac {N^2-1}{2}-\tfrac {11N}{6})\tfrac{\alpha(1-\alpha)}{2^N} + ((\tfrac{N^2}{2}-\tfrac{5N}{6})^2-\tfrac{1+3N^2}{2})\tfrac{(1-\alpha)^2}{2^{2N}}\\
&\geq 0 \quad \text{ for } N\geq 5,
\end{align*}
which implies that $D-\tfrac \epsilon2\id\geq 0$. This completes the proof for the odd case, and the proof of Lemma \ref{lemG} is complete.

\end{document}